\documentclass[pra, twocolumn]{revtex4-1}
\usepackage{amsmath}
\usepackage{amsfonts}
\usepackage{amssymb}
\usepackage{bbm}
\usepackage{graphicx}
\usepackage[caption=false]{subfig}
\usepackage{array}
\usepackage{hyperref}
\usepackage{xcolor}
\usepackage{color}
\usepackage[normalem]{ulem}

\providecommand{\U}[1]{\protect\rule{.1in}{.1in}}

\newtheorem{theorem}{Theorem}[section]

\newtheorem{definition}[theorem]{Definition}

\newtheorem{lemma}[theorem]{Lemma}
\newtheorem{proposition}[theorem]{Proposition}
\newtheorem{solution}[theorem]{Solution}

\newenvironment{proof}[1][Proof]{\noindent\textbf{#1:} }{\ \rule{0.5em}{0.5em}}

\newcommand{\ket}[1]{| #1 \rangle}

\DeclareMathOperator{\Tr}{Tr}

\begin{document}

\title{Randomized Graph States and their Entanglement Properties}
\author{Jun-Yi Wu$^{1}$, Matteo Rossi$^{2}$, Hermann Kampermann$^{1}$,
Simone Severini$^{3}$, Leong Chuan Kwek$^{4}$, Chiara Macchiavello$^{2}$,
Dagmar Bru\ss $^{1}$}

\affiliation{$^{1}$Institut f\"ur Theoretische Physik III,
Heinrich-Heine-Universit\"at D\"usseldorf,
D-40225 D\"usseldorf, Germany}

\affiliation{$^{2}$Dipartimento di Fisica and
INFN-Sezione di Pavia, via Bassi 6, 27100 Pavia, Italy}

\affiliation{$^{3}$Department of Computer Science, University College London,
Gower Street, London WC1E 6BT, United Kingdom}

\affiliation{$^{4}$Centre for Quantum Technologies, National University of
Singapore, 3 Science Drive 2,
Singapore 117543, Singapore}

%\date[Created Date: ]{\today }

\begin{abstract}
We introduce a class of mixed multiqubit states, that corresponds to a randomized version of graph states. Such states arise when a graph state is prepared with noisy or imperfect controlled-$Z$ gates. We study the entanglement features of these states by investigating both bipartite and genuine multipartite entanglement. Bipartite entanglement is studied via the concepts of connectedness and persistency, which are related to measurement based quantum computation. The presence of multipartite entanglement is instead revealed by the use of witness operators which are subsequently adapted to study nonlocal properties through the violation of suitable Bell inequalities.
We also present results on the entanglement detection of particular randomized graph states, by deriving explicit thresholds for entanglement and nonlocality in terms of the noise parameter that characterizes the controlled-$Z$ gates exploited for their generation. Finally, we propose a method to further improve the detection of genuine multipartite entanglement in this class of states.

\end{abstract}

\maketitle

%\tableofcontents

%%%%%%%%%%%%%%%%%%%%%%%%%%%%%%%%%%%%%%%%%%%%%%%%%%%%%%%%%%%%%%%%%%%%%%%%%%%%%%%%%%%%%%%%%%%%%%%%%%%%%%%%%%%%%%
%%   Introduction
%

\section{Introduction}

\label{sec::Introduction} % About graph states

Graph states and especially cluster states are at the heart of measurement
based quantum computation (MBQC) \cite{RaussBrie2001-5}. Given a cluster
state, this prominent model of quantum computation provides a way to perform
universal computing with only local gates and measurements, by avoiding the
use of two-qubit entangling gates. Under this light, the entanglement
content of cluster states can then be regarded as a quantum resource that is
consumed throughout the process. However, despite the fact that all the
operations involved in MBQC can nowadays be easily implemented in various
hardware, the hardest task from an experimental point of view is represented
by the preparation of the initial cluster state.

% About graph states preparation
The preparation of general graph states always starts from a product state
of qubits corresponding to the vertices of a graph with no edges, which is
then subsequently processed via an Ising-like interaction \cite%
{BriegelRaussendorf2001-01}. This interaction is tuned in such a way that
its action can be regarded as a series of controlled-$Z$ (CZ) gates,
connecting the vertices according to the target graph. In Ref. \cite%
{Cabello2011-04} a preparation method involving only one- and two-qubit
gates for graph states up to $12$ qubits is proposed. As a matter of fact,
the current experimental realization of a CZ gate is far from being perfect,
and in practice it is very difficult to create a noiseless graph state \cite%
{Cabello2011-04}.

% About imperfect graph states preparation
A possible way to model a noisy CZ gate is to assume that, with probability $%
p$ it creates the desired edge between its qubits, while with probability $%
1-p$ it fails. For heralded entanglement \cite{lim2005repeat}, if the gate fails, one could recover the original state, i.e. $|++\rangle$. This has the same effect as an identity operator. A physical realization of this
probabilistic CZ gate was suggested in \cite{lim2005repeat,lim2006repeat,beige2007repeat}.

%%%%%%%% statements
% statement 1, important per se, ent:
In this paper, we aim at studying the \emph{randomized graph state} (for short, RG state), that is, states that arise whenever a probabilistic CZ gate is applied for every edge in a graph.
Given a graph state, its randomized version is thus a mixture of all the states corresponding to
its subgraphs. These are weighted according to a single parameter $p$, which we
call \emph{randomness parameter}, physically related to the success
probability of the CZ gate.

\par

Besides addressing the issue of the
unitary equivalence of general RG states, we will mainly focus
on the amount of entanglement in RG states, both in the bipartite and the
multipartite case \cite{AcinBLSanpera2001-07}. Regarding the former, we will
especially discuss the concepts of persistency and connectedness, which have
a clear application in terms of the usefulness of RG states for MBQC \cite%
{BriegelRaussendorf2001-01}. For the quantification of the latter, we will
use a genuine multipartite entanglement witness \cite%
{Bourennane2004-02,GuhneHBELMSanpera2002-12,TothGuhne2005-02}. We will be
able in this way to define a critical value $p_{c}$ for the randomness
parameter, above which the state shows genuine multipartite entanglement
properties. Finally, nonlocal realistic features of RG states will be
discussed with the help of suitable Bell inequalities developed for graph
states.

% statement 2, imperfection and other uses possible:
Notice that, not only are RG states interesting and highly non trivial
\textit{per se}, but they are a useful tool to investigate and understand
the presence of noise in MBQC. Furthermore, complete RG states are a
plausible quantum counterpart to the classical Erd\H{o}s-R{\'{e}}nyi random
graphs introduced in \cite{ErdosRenyi1959} (Ref. \cite{janson2000_randomgraphs} is
a detailed survey on the topic), and recently studied in the context of complex systems \cite{Newman2003,
BoccalettiLMCHwange2006}.

%%%%%%%% stucture
The present paper is organized as follows. In Sec. \ref{sec::Preliminary} we review some basic definitions about mathematical graphs, random graphs and quantum graph states. We define randomized graph states in Sec. \ref{sec::randomized_graph_state}. We then study the rank of RG states to answer the question of unitary equivalence and bipartite and multipartite entanglement in Secs. \ref{sec::unit-eq}, \ref{sec::bipartite_entanglement}, and \ref{sec::GME}, respectively. In Sec. \ref{sec::GME}, an approximation to a witness for multipartite entanglement is introduced, which allows to determine a threshold probability. A further analysis on nonlocal realism is carried out in Sec. \ref{sec:bell}. We conclude in Sec. \ref{sec::conclusion} with a summary of the achieved results and future perspectives.

\section{Preliminaries}

\label{sec::Preliminary}

In this section, we briefly review the definition of graphs as used in the
paper and the mathematical concept of Erd\H{o}s-R{\'{e}}nyi random graphs.
We then remind the reader of the well-known class of quantum graph states
and introduce the notation that will be used throughout the paper.
%The readers that are already used to these topics can just skip to the definition of random graph states in the section \ref{sec::randomized_graph_state}.

%%%%%%%%%%%%%%%%%%%%%%%%%%%%%%%%%%%%%%%%%%%%%%%%%%%%%%%%%%%%%%%%%%%%%%%%%%%%%%%%%%%%%%%%%%%%%%%%%%%%%%%%%%%%%%
%%   Subsection "Graphs"

\subsection{Graphs}
\label{sec::graphs}

A \emph{graph} $G=(V,E)$ is defined as a pair consisting of a
set $V_{G}=\left\{ v_{1},\cdots ,v_{n}\right\} $, whose elements are called
\emph{vertices, }and a set $E_{G}=\left\{ e_{1},\cdots ,e_{l}\right\} $,
whose elements are called \emph{edges }and consist of unordered pairs of
different vertices \cite{Diestel_GraphTheory}. A graph $F$ with $%
V_{F}\subseteq V_{G}$ and $E_{F}\subseteq E_{G}$ is called a \emph{subgraph}
of $G$. If $V_{F}=V_{G}$ then $F$ is said to be a \emph{spanning subgraph}
of $G$; in such a case, we say that $F$ \emph{ spans } $G$. Two vertices are
\emph{neighbors} if they are connected by an edge. The \emph{degree} of a
vertex $v_{i}$, $d_{v_{i}}$, is the number of its neighbors. A graph is
\emph{empty} if it has no edges. The empty graph on $n$ vertices is denoted
by $G_{n}^{\emptyset }$. On the other hand, the \emph{complete} (or fully
connected) graph on $n$ vertices, $K_{n}$, contains all possible $\binom{n%
}{2}$ edges. Other relevant types of graphs that will be considered along
the paper are the following ones:

\begin{itemize}
\item[-] \emph{Star graphs}, $S_{n}$: graphs where one vertex has degree $n-1$ and all others have degree $1$.

\item[-] \emph{Cluster graphs}, $L_{m\times n}$: graphs whose vertices correspond to the points of a discrete two-dimensional lattice with $m$ times $n$. When $m=1$, we simply write $L_{n}$. This is a \emph{linear cluster}, or, equivalently, a path on $n$ vertices. Notice that in the graph-theoretic literature $L_{m\times n}$ is usually called a grid graph or a lattice graph. We use a different terminology given the link with MBQC.

\item[-] \emph{Cycle graphs}, $C_{n}$: graphs where all vertices have degree $2$. These are closed linear clusters.
\end{itemize}

A very useful concept in the remainder of the paper is the symmetric difference. Letting $F$ and $G$ be two graphs on the same set of vertices $V$, their \emph{symmetric difference} is the graph $F\Delta G$, such that $V_{F\Delta G}=V_{F}=V_{G}$ and $E_{F\Delta G}=E_{F}\cup E_{G} \setminus E_{F}\cap E_{G} $.

%%   End of subsection "Graphs"
%%%%%%%%%%%%%%%%%%%%%%%%%%%%%%%%%%%%%%%%%%%%%%%%%%%%%%%%%%%%%%%%%%%%%%%%%%%%%%%%%%%%%%%%%%%%%%%%%%%%%%%%%%%%%%

%%%%%%%%%%%%%%%%%%%%%%%%%%%%%%%%%%%%%%%%%%%%%%%%%%%%%%%%%%%%%%%%%%%%%%%%%%%%%%%%%%%%%%%%%%%%%%%%%%%%%%%%%%%%%%%
    %%   Subsection "Random Graphs"
    %

\subsection{Erd\H{o}s-R{\'e}nyi random graphs}

Random graphs are a well-developed mathematical subject touching both graph theory and probability theory {\textbf{\ }}\cite{janson2000_randomgraphs}. In the Erd\H{o}s-R{\'{e}}nyi (ER) random graph on $n$ vertices, each edge is
included with probability $p$ independently of any other edge. Notice that, as $p$ is uniform for all edges, then the probability of a subgraph $G\subseteq K_{n}$ with a number of edges $|E_{G}|$ is given by $P(G)=p^{|E_{G}|}(1-p)^{\binom{n}{2}-|E_{G}|}$. As an illustration, Fig. \ref{fig::random_subgraphs} shows all possible subgraphs of the complete graph $K_{3}$.

\begin{figure}[t!]
\subfloat[The empty $G^{\emptyset}_3$ and the complete $K_{3}$ subgraphs with probability $(1-p)^3$ and $p^3$, respectively.]{
\includegraphics[width=0.8\linewidth]{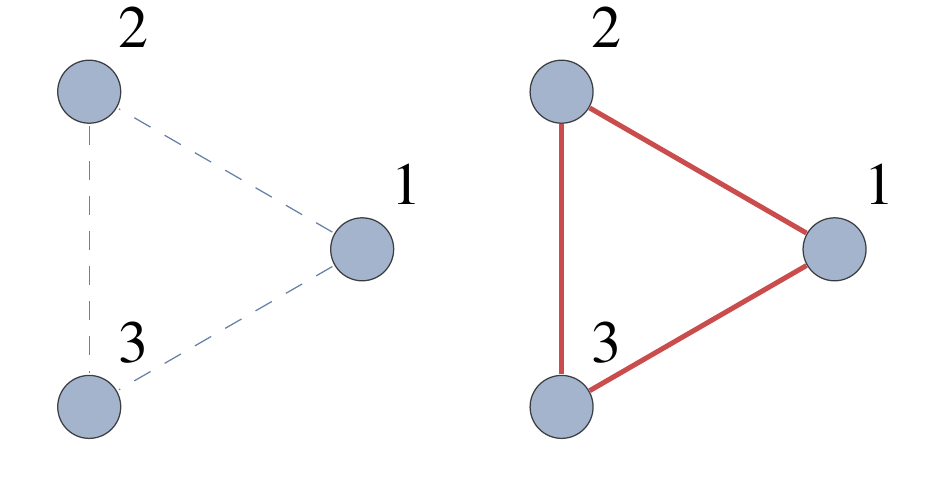}
\label{fig::random_subgraphs_01}
}
\newline
\subfloat[The subgraphs composed of a single edge with probability $p(1-p)^2$.]{
\includegraphics[width=\linewidth]{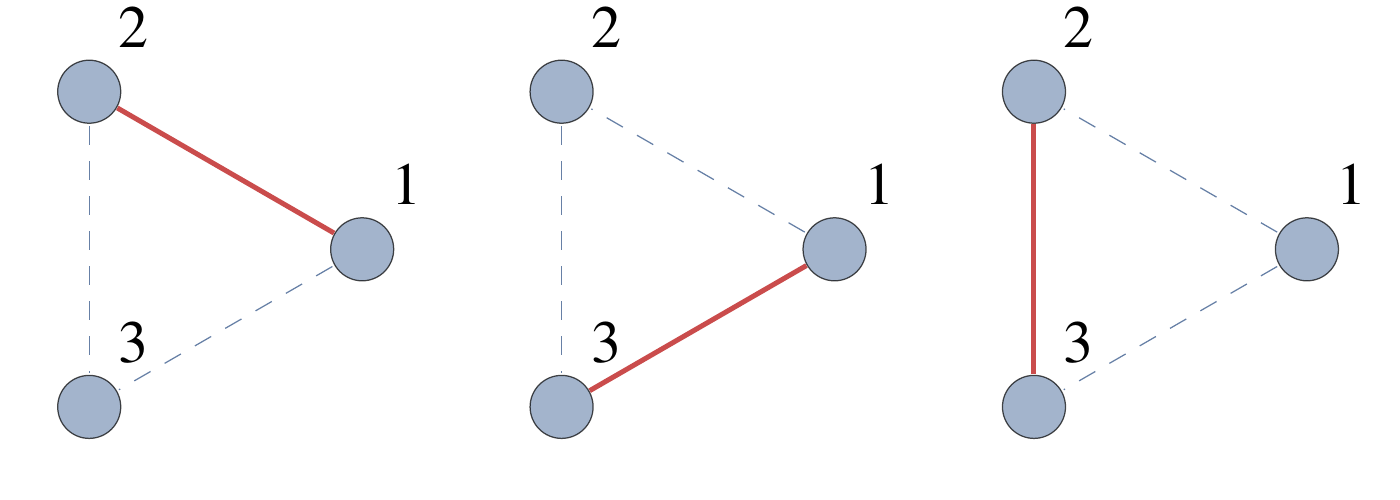}
\label{fig::random_subgraphs_02}
}
\newline
\subfloat[The subgraphs composed of two edges with probability $p^2(1-p)$.]{
\includegraphics[width=\linewidth]{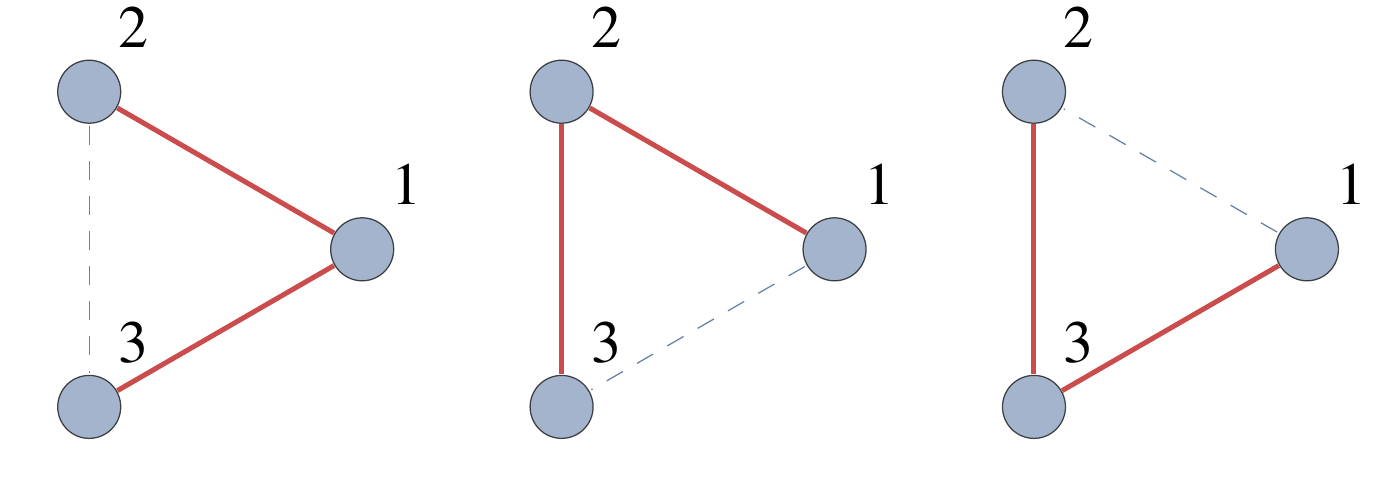}
\label{fig::random_subgraphs_03}
}
\caption{(Color online) All possible subgraphs on three vertices, and the
related probabilities, as instances of the ER random graph.}
\label{fig::random_subgraphs}
\end{figure}
%%   End of subsection "Random Graphs"
%%%%%%%%%%%%%%%%%%%%%%%%%%%%%%%%%%%%%%%%%%%%%%%%%%%%%%%%%%%%%%%%%%%%%%%%%%%%%%%%%%%%%%%%%%%%%%%%%%%%%%%%%%%%%%

%%%%%%%%%%%%%%%%%%%%%%%%%%%%%%%%%%%%%%%%%%%%%%%%%%%%%%%%%%%%%%%%%%%%%%%%%%%%%%%%%%%%%%%%%%%%%%%%%%%%%%%%%%%%%%
%%   Subsection "Graph States"

\subsection{Graph states}

\label{sec::graph_state}

We will briefly review here the well known concept of a graph state of $n$
qubits and its connection to graphs \cite{HeinEisertBriegel2004-06, hein2006entanglement}. Given a graph $G=(V,E)$ on $n$
vertices, the corresponding graph state is denoted by $|G\rangle $ and
defined as follows. First, assign to each vertex a qubit and initialize it
as the state $|+\rangle =\frac{1}{\sqrt{2}}(|0\rangle +|1\rangle )$, so that
the initial $n$-qubit state is given by $|+\rangle ^{\otimes n}$. Then,
perform a CZ operation between any two qubits
associated to vertices that are connected by an edge. This operation is defined
as $\mathrm{CZ}=\mathrm{diag}(1,1,1,-1)$, in the computational basis $\{|0\rangle
,|1\rangle \}$ for each qubit. By performing the CZ operation on any two
connected qubits $i_{1}$ and $i_{2}$, we get the corresponding graph state
\begin{equation}
|G\rangle :=\prod_{\{i_{1},i_{2}\}\in E}(\text{CZ})_{i_{1}i_{2}}|+\rangle
^{\otimes n}.  \label{eq.::def_graphstate}
\end{equation}%
Notice that the number of distinct graph states of $n$ qubits is equal to $%
2^{\binom{n}{2}}$, which is the number of labeled graphs with $n$ vertices.
%%   End of subsection "Graph States"
%%%%%%%%%%%%%%%%%%%%%%%%%%%%%%%%%%%%%%%%%%%%%%%%%%%%%%%%%%%%%%%%%%%%%%%%%%%%%%%%%%%%%%%%%%%%%%%%%%%%%%%%%%%%%%

\section{Randomized Graph States}

\label{sec::randomized_graph_state}

In this section, we will introduce the class of randomized graph (RG) states.
The main idea is to start from a graph $G$ and to apply probabilistic gates $%
\Lambda _{p}$ to the state $|+\rangle ^{\otimes n}$ instead of the perfect
CZ gates. $\Lambda _{p}$ is defined as
\begin{equation}
\Lambda _{p}(|++\rangle \langle ++|)=p|\vcenter{\hbox{%
\includegraphics[width=2em]{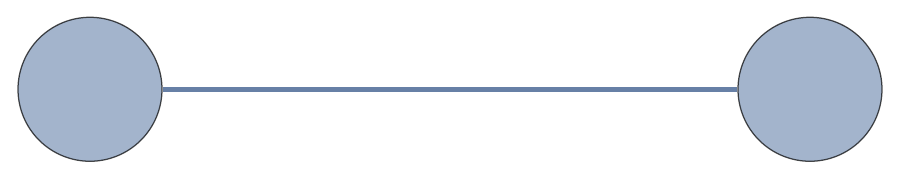}}}\rangle \langle %
\vcenter{\hbox{\includegraphics[width=2em]{pics/bell_2}}}|+(1-p)|++\rangle
\langle ++|,  \label{eq.::probabilistic_gate_random_graph_model}
\end{equation}%
with $|++\rangle $ representing the two-qubit empty graph state, and $|%
\vcenter{\hbox{\includegraphics[width=2em]{pics/bell_2}}}\rangle $ denoting
the two-qubit connected graph state. In other words, we consider a noisy
implementation of the gate CZ, where one realizes the desired CZ gate with
probability $p$, but one fails and does nothing with probability $1-p$ \cite%
{lim2005repeat,lim2006repeat,beige2007repeat}. Notice that all gates $%
\Lambda _{p}$ acting on any pair of qubits commute and therefore we do not
have to specify the order of application.

As an illustration, suppose we want to generate the GHZ state $|%
\vcenter{\hbox{\includegraphics[width=2em]{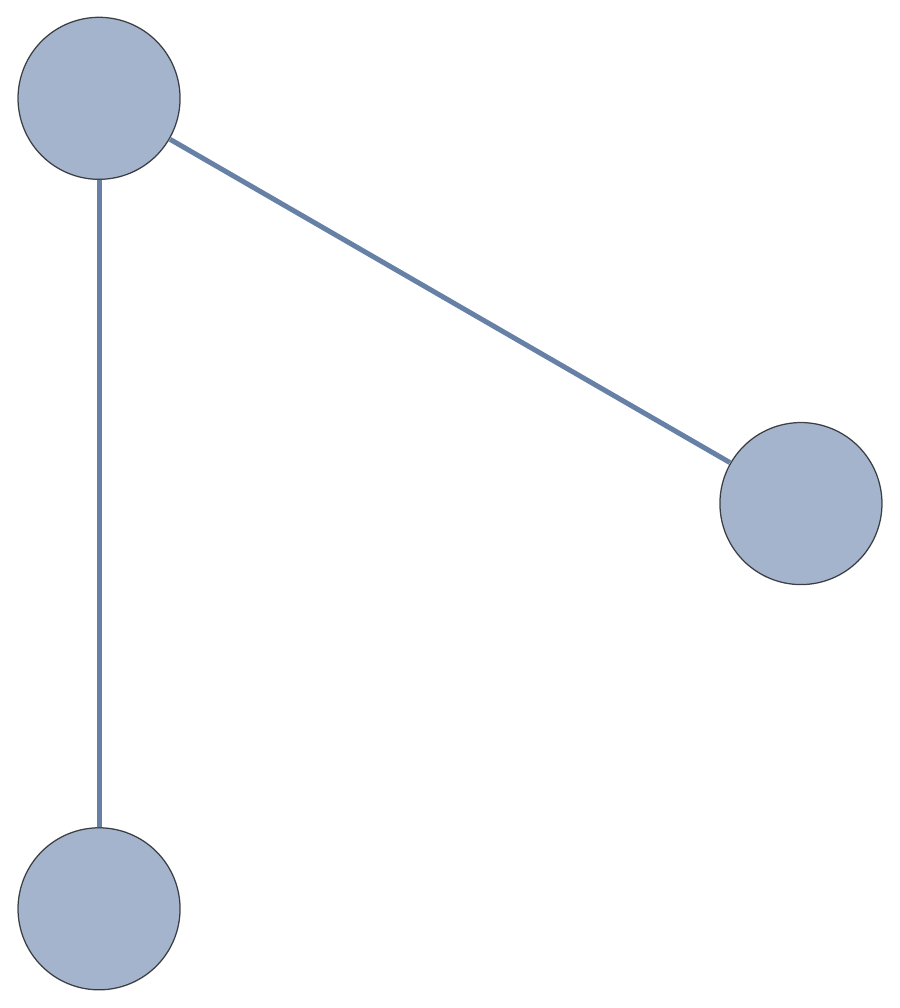}}}\rangle $ by
employing the aforementioned procedure, namely by applying the probabilistic
gates $\Lambda _{p}$ to create edges. It is easy to see that the resulting
state is a mixture of subgraph states of $|\vcenter{\hbox{%
\includegraphics[width=2em]{pics/GHZ_3}}}\rangle $, namely
\begin{align}
R_{p}(|\vcenter{\hbox{\includegraphics[width=2em]{pics/GHZ_3}}}\rangle) &
=\Lambda _{p}^{\{1,2\}}\circ \Lambda _{p}^{\{2,3\}}\left( |+++\rangle
\langle +++|\right)  \nonumber\\
& =p^{2}|\vcenter{\hbox{\includegraphics[width=2em]{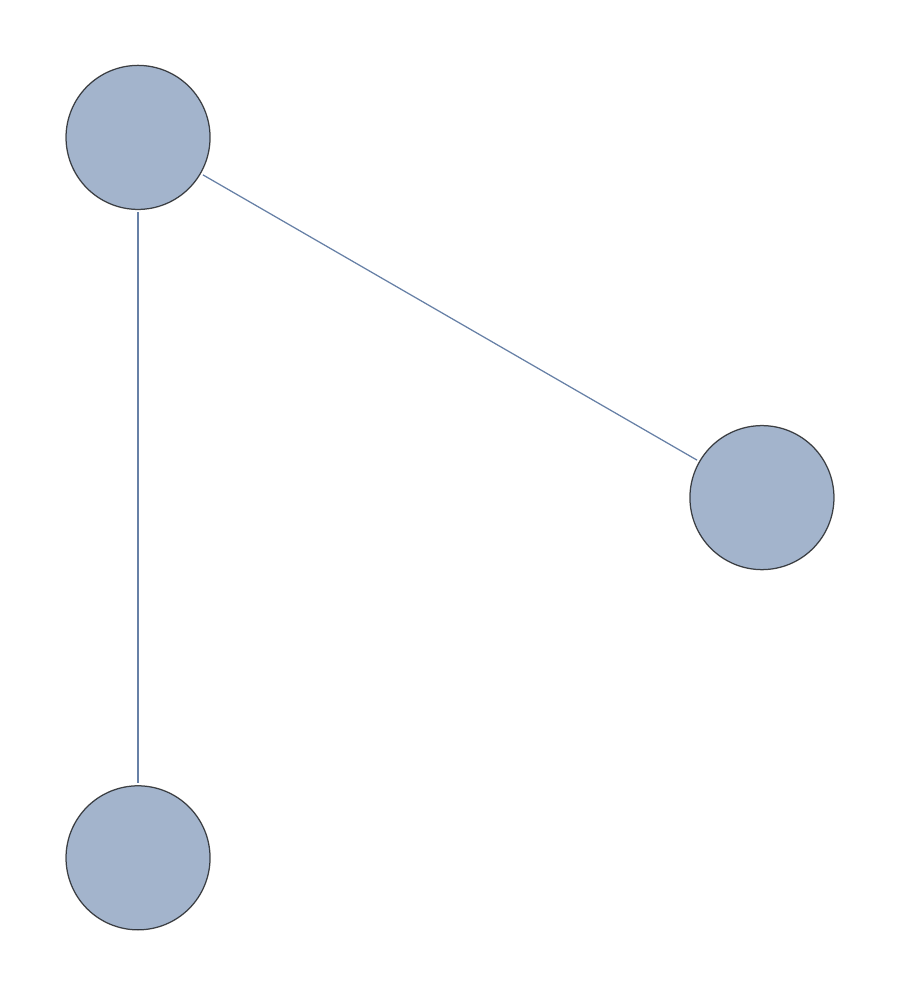}}}%
\rangle \langle \vcenter{\hbox{%
\includegraphics[width=2em]{pics/Star3_subgraph_4_no_label}}}|  \notag \\
& +p(1-p)|\vcenter{\hbox{\includegraphics[width=2em]{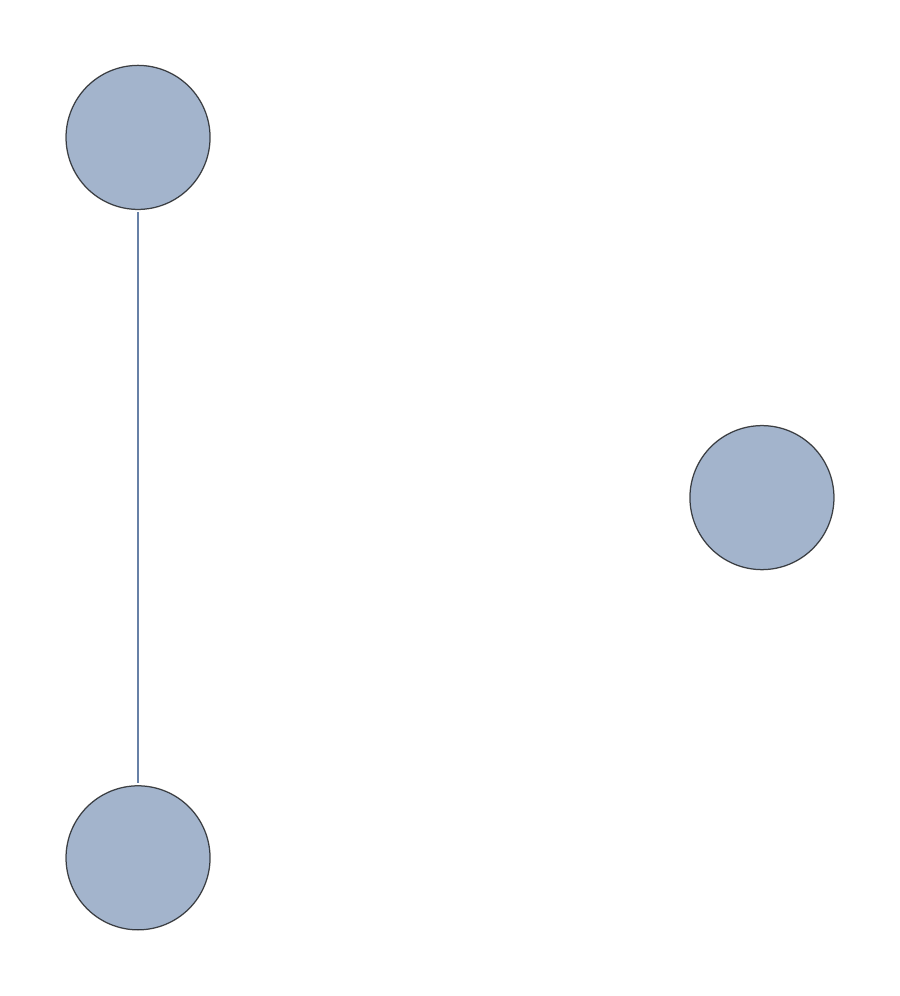}}}%
\rangle \langle \vcenter{\hbox{%
\includegraphics[width=2em]{pics/Star3_subgraph_3_no_label}}}|+p(1-p)|%
\vcenter{\hbox{\includegraphics[width=2em]{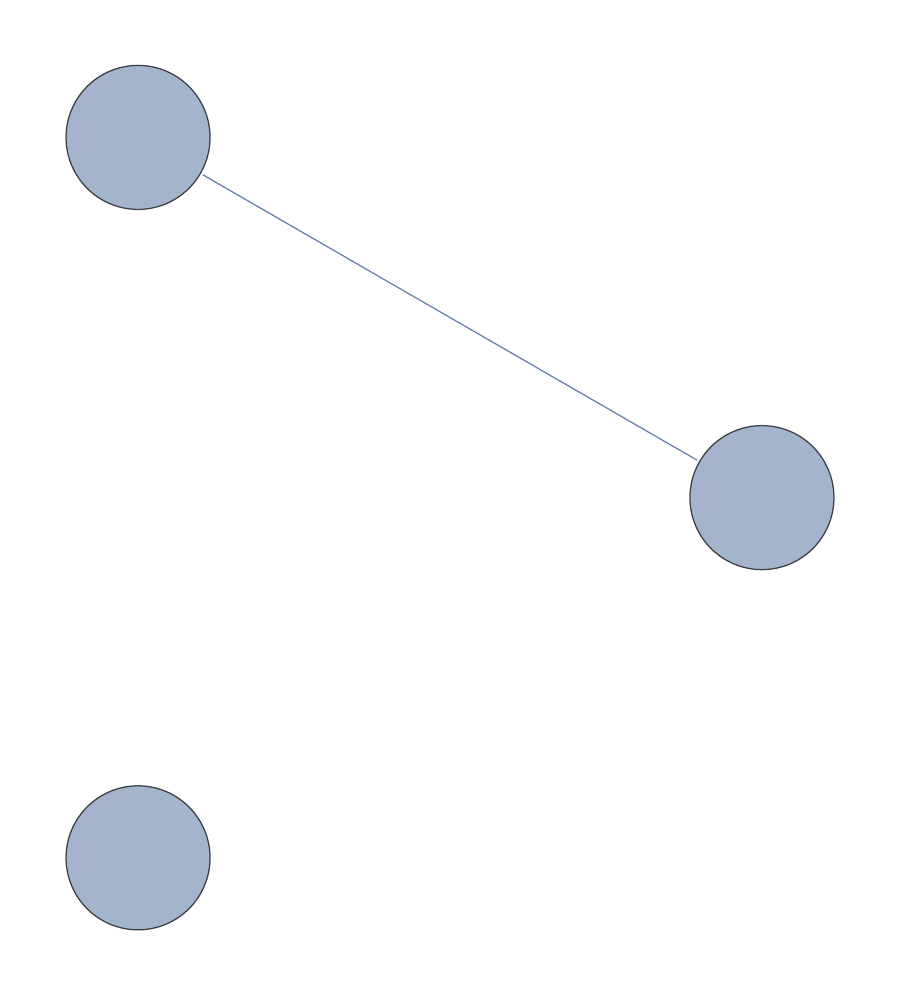}}}\rangle
\langle \vcenter{\hbox{\includegraphics[width=2em]{pics/Star3_subgraph_2_no_label}}}|
\notag \\
& +(1-p)^{2}|\vcenter{\hbox{%
\includegraphics[width=2em]{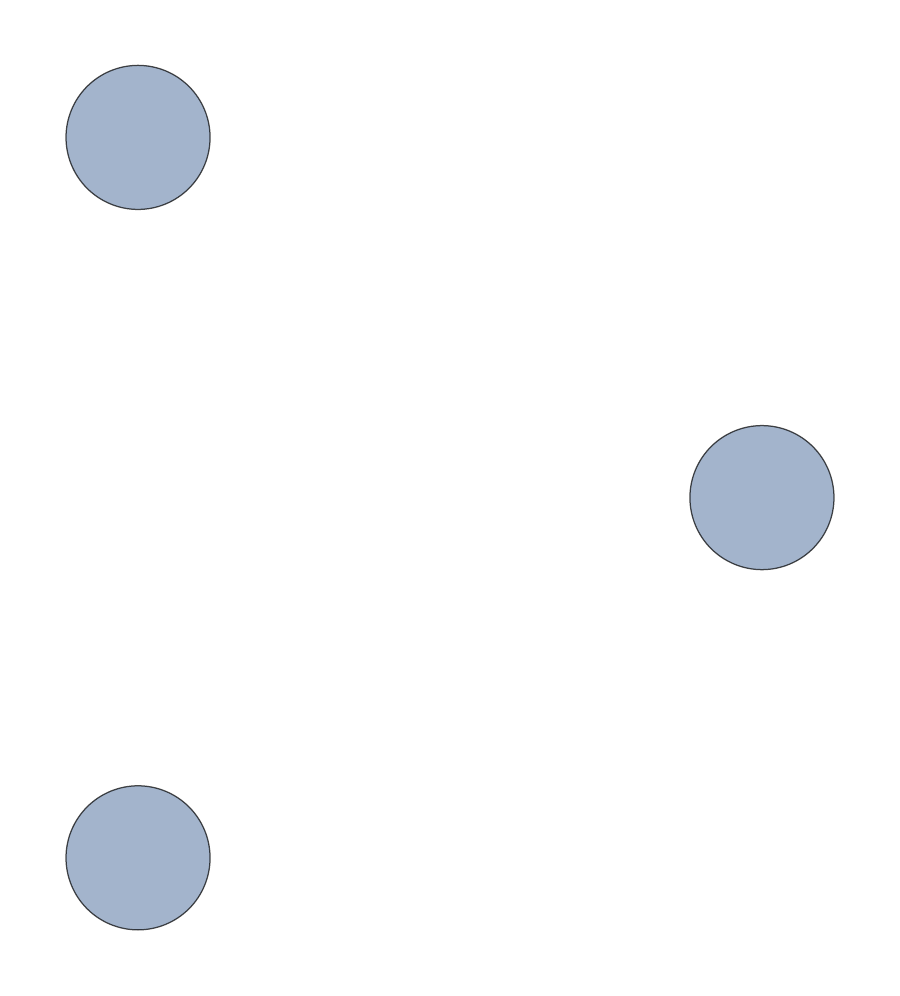}}}\rangle \langle %
\vcenter{\hbox{\includegraphics[width=2em]{pics/Star3_subgraph_1_no_label}}}|.
\label{eq.::RS3_state}
\end{align}%
The above state is then said to be the RG state associated to the graph $%
\vcenter{\hbox{\includegraphics[width=2em]{pics/GHZ_3}}}$. The above example
shows that the RG state $\rho _{G}^{p}$ associated to a graph $G$, or
equivalently to a pure graph state $|G\rangle $, can be derived by applying
the randomization operation $R_{p}$ in agreement with the following
definition.

\begin{definition}[Randomized graph state]
\label{def.::randomization_operator_random_graph_state}Let $|G\rangle $ be a
graph state. A \emph{randomization operator} $R_{p}$ is defined via%
\begin{equation}
R_{p}(|G\rangle ):=\sum_{F \emph{spans} G}p^{|E_{F}|}\left( 1-p\right)
^{\left\vert E_{G}\backslash E_{F}\right\vert }|F\rangle \langle F|,
\label{eq.::def_randomization_operator_random_graph_state}
\end{equation}%
where $F$ are spanning subgraphs of $G$, $E_{F}$ and $E_{G}$ are the sets of
edges of $F$ and $G$, and $p$ is the randomness parameter
corresponding to the success probability of the CZ gate in Eq. %
\eqref{eq.::probabilistic_gate_random_graph_model}. The resulting state $%
\rho _{G}^{p}:=R_{p}(|G\rangle )$ is the randomized version of $|G\rangle $
with randomness parameter $p$, or, shortly, a $p$\emph{-randomization} of $%
|G\rangle $.
\end{definition}

This randomization operator corresponds to the preparation of graph states
showed in the probabilistic gate model of Eq. %
\eqref{eq.::probabilistic_gate_random_graph_model}. It maps a pure graph
state $|G\rangle $ into a mixture of all its spanning subgraph states. Since
the two extreme cases $p=0,1$ correspond to the empty graph and the pure
graph state, respectively, the parameter $p$ plays a fundamental role to
determine the entanglement features of RG states.

In addition, it is useful to remark a difference between mathematical ER
random graphs and RG states:\ in ER random graphs all possible edges among
the vertices are considered; in RG states the randomization is restricted to
the edges of a given graph. In other words, ER random graphs are always
related to the fully connected graph, while RG states can be generated by
the randomization process on any graph. From this viewpoint, we can say that
RG states are more general than random graphs, since only in the case of $%
G=K_{n}$ does the corresponding RG state $\rho _{K_{n}}^{p}$ have the same
combinatorial properties as the ER random graph of $n$ vertices. It
is then evident that our model is in close analogy with bond percolation. Of
course, the questions that we ask are not directly related to the main
question in percolation theory, which is traditionally concerned with the
global behavior of infinite graphs as a function of the randomness parameter
(see \cite{Grimmett1999-percolation}).

In this paper we will denote the $p$ randomization of the important graph
states $|K_{n}\rangle $, $|S_{n}\rangle $, $|L_{n}\rangle $, and $%
|C_{n}\rangle $ by $\rho _{K_{n}}^{p}$, $\rho _{S_{n}}^{p}$, $\rho
_{L_{n}}^{p}$, and $\rho _{C_{n}}^{p}$, respectively.
\par

Notice that a different definition of random graph states is also given in
\cite{CollinsNechitaZyczkowski2010}. In that model, a vertex with degree $d$
is represented by a $d$-qubit system and two vertices $a$ and $b$ are said to
be connected by an edge if one qubit in $a$ is maximally entangled with one
qubit in $b$. A random unitary matrix describes the coupling between
subsystems of a vertex. The random graph states considered in \cite%
{CollinsNechitaZyczkowski2010} are then an ensemble of pure states. In
contrast, in our definition each vertex is a single-qubit system, and a
randomized graph state is always a mixed state for any value of the
randomness parameter $0<p<1$. Notice that other ways to define
mixed quantum states from graphs have been studied in the literature (see
e.g. Ref. \cite{braunstein2006laplacian}).

\section{Rank of randomized graph states and unitary equivalence}
\label{sec::unit-eq}

In this section, we investigate the question of \emph{%
local unitary} (LU) equivalence of RG states. Two $n$-qubit quantum states $%
\rho $ and $\sigma $ are \emph{LU equivalent} if and only if there exist
local unitaries $U^{(1)},...,U^{(n)}$ such that $\rho =U^{(1)}\otimes \cdots
\otimes U^{(n)}\sigma U^{(1)\dagger }\otimes \cdots \otimes U^{(n)\dagger }$%
. LU equivalent states have identical entanglement properties.

The LU equivalence classes of graph states have been intensively studied in
Ref. \cite{HeinEisertBriegel2004-06}. Pure graph states up to six qubits can
be classified in $19$ different LU classes. Graph states in the same class
can be transformed into each other via local unitaries, and hence share the
same entanglement properties. However, in most cases the RG states derived
from two LU equivalent graph states, say $|G_{1}\rangle $ and $|G_{2}\rangle
$, are not LU equivalent and, in general, not even equivalent under \emph{%
global unitaries} (GU).

In order to see this, consider for instance the graph states $|G_{1}\rangle
=|\vcenter{\hbox{\includegraphics[width=2em]{pics/GHZ_3}}}\rangle $ and $%
|G_{2}\rangle =|\vcenter{\hbox{\includegraphics[width=2em]{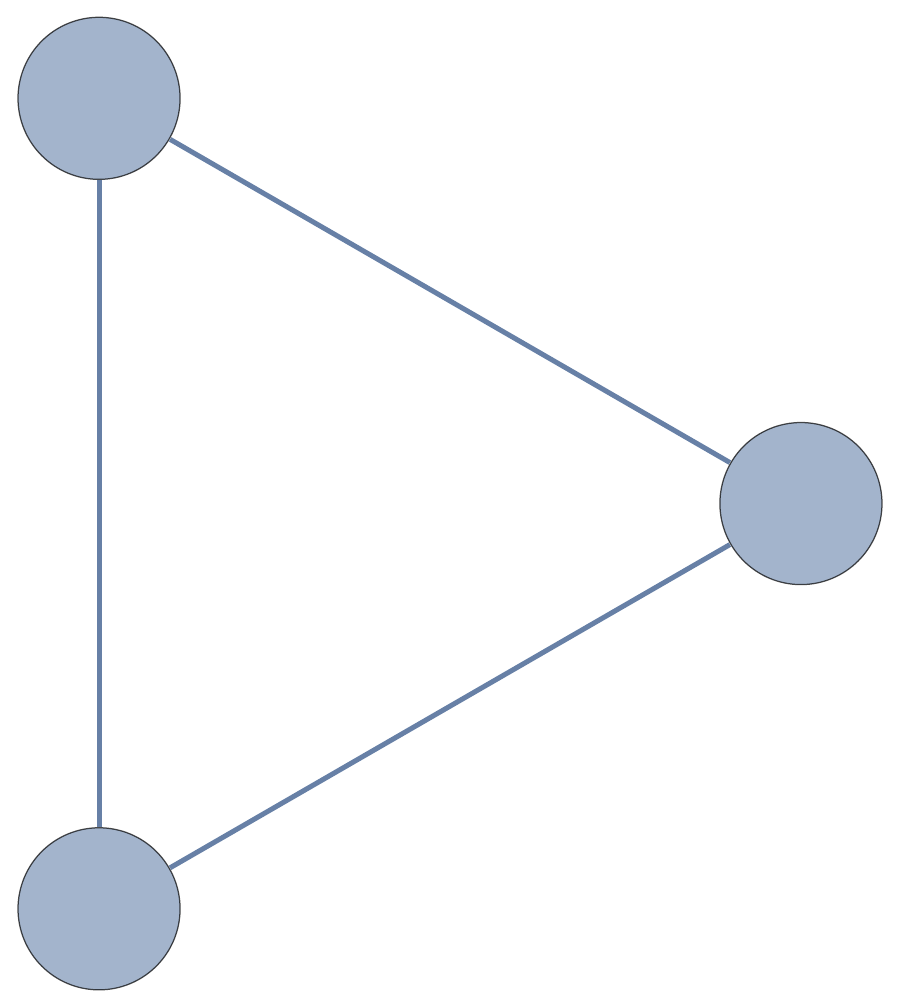}}}%
\rangle $, that are known to be LU equivalent. The corresponding RG states
are given by $\rho _{G_{1}}^{p}=R_{p}(|\vcenter{\hbox{%
\includegraphics[width=2em]{pics/GHZ_3}}}\rangle )$, see Eq. %
\eqref{eq.::RS3_state}, and
\begin{align}
\rho _{G_{2}}^{p}& =p^{3}|\vcenter{\hbox{%
\includegraphics[width=2em]{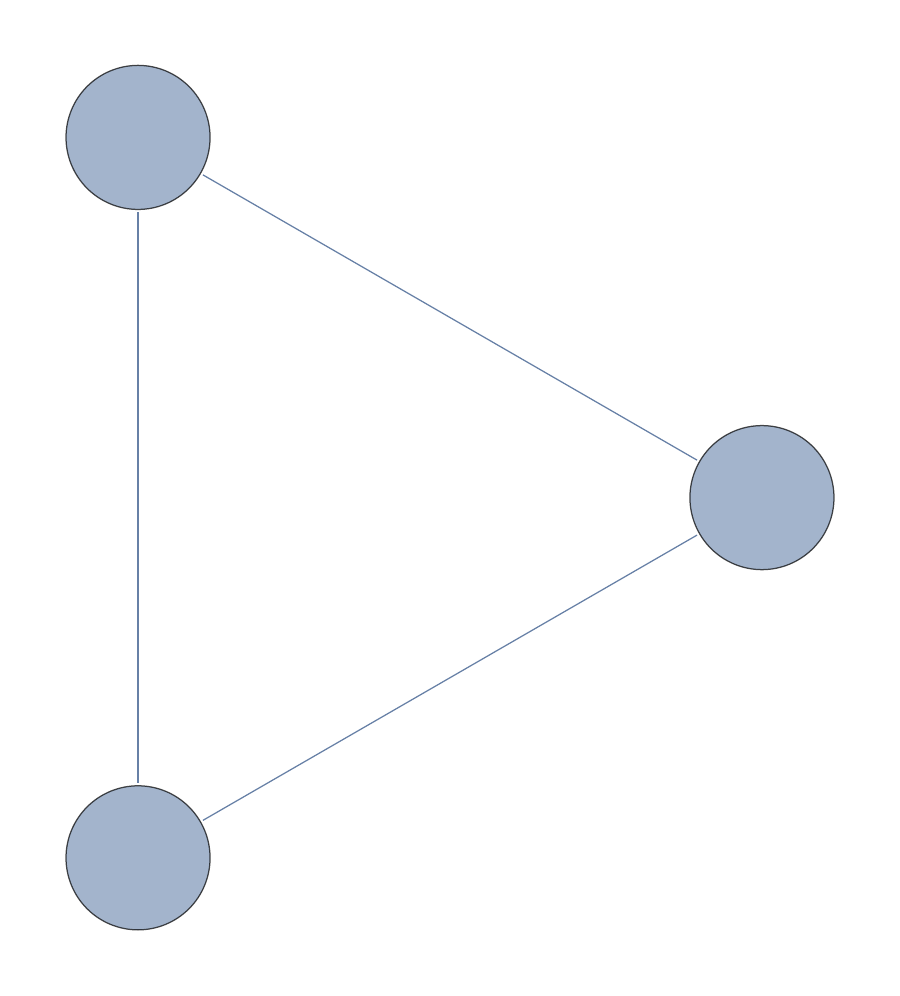}}}\rangle \langle %
\vcenter{\hbox{\includegraphics[width=2em]{pics/K3_subgraph_3_edges_no_label}}}|  \notag
\label{eq.::RC3_state} \\
& +p^{2}\left( 1-p\right) |\vcenter{\hbox{%
\includegraphics[width=2em]{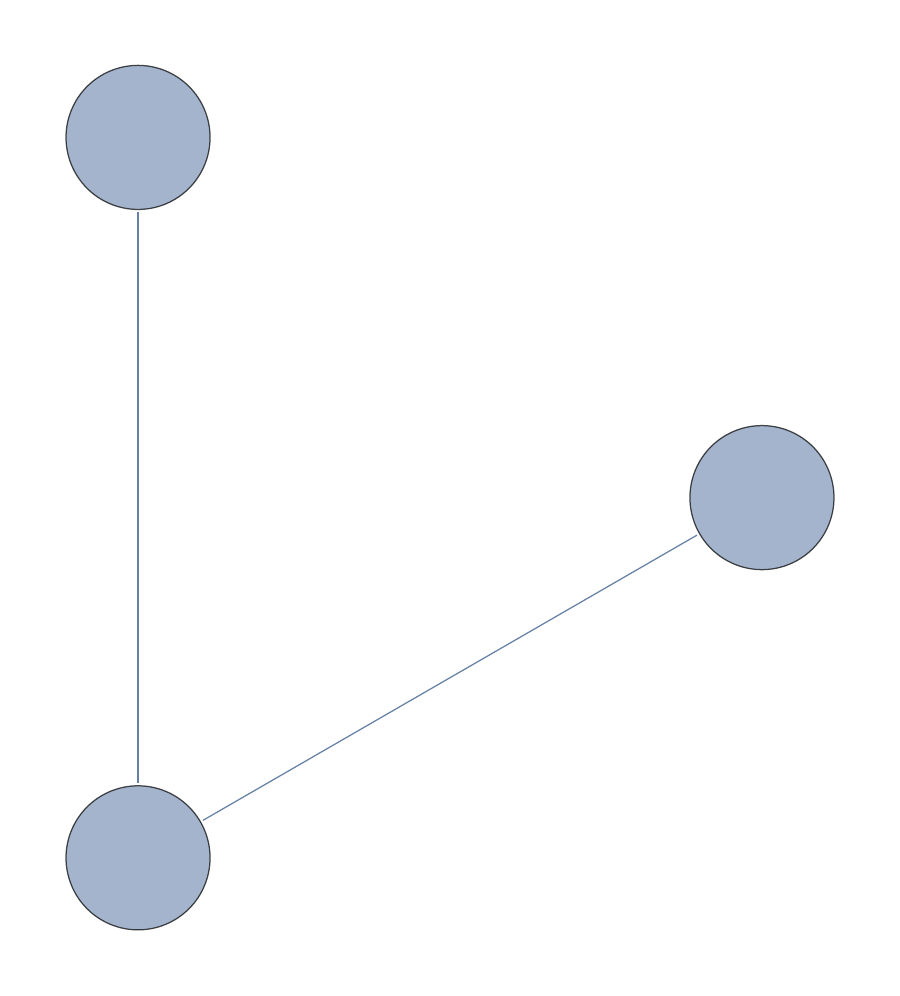}}}\rangle \langle %
\vcenter{\hbox{\includegraphics[width=2em]{pics/K3_subgraph_2_edges_no_label}}}|+\cdots
\notag \\
& +p\left( 1-p\right) ^{2}|\vcenter{\hbox{%
\includegraphics[width=2em]{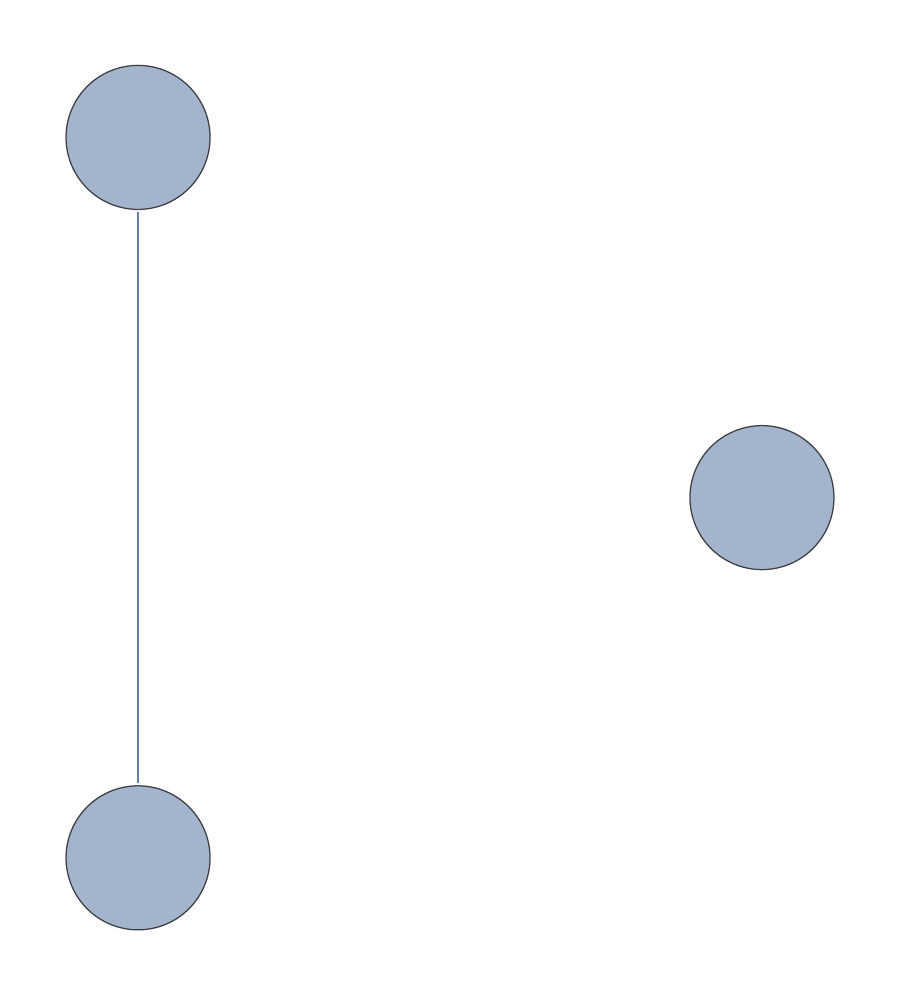}}}\rangle \langle %
\vcenter{\hbox{\includegraphics[width=2em]{pics/K3_subgraph_1_edges_no_label}}}|+\cdots
\notag \\
& +\left( 1-p\right) ^{3}|\vcenter{\hbox{%
\includegraphics[width=2em]{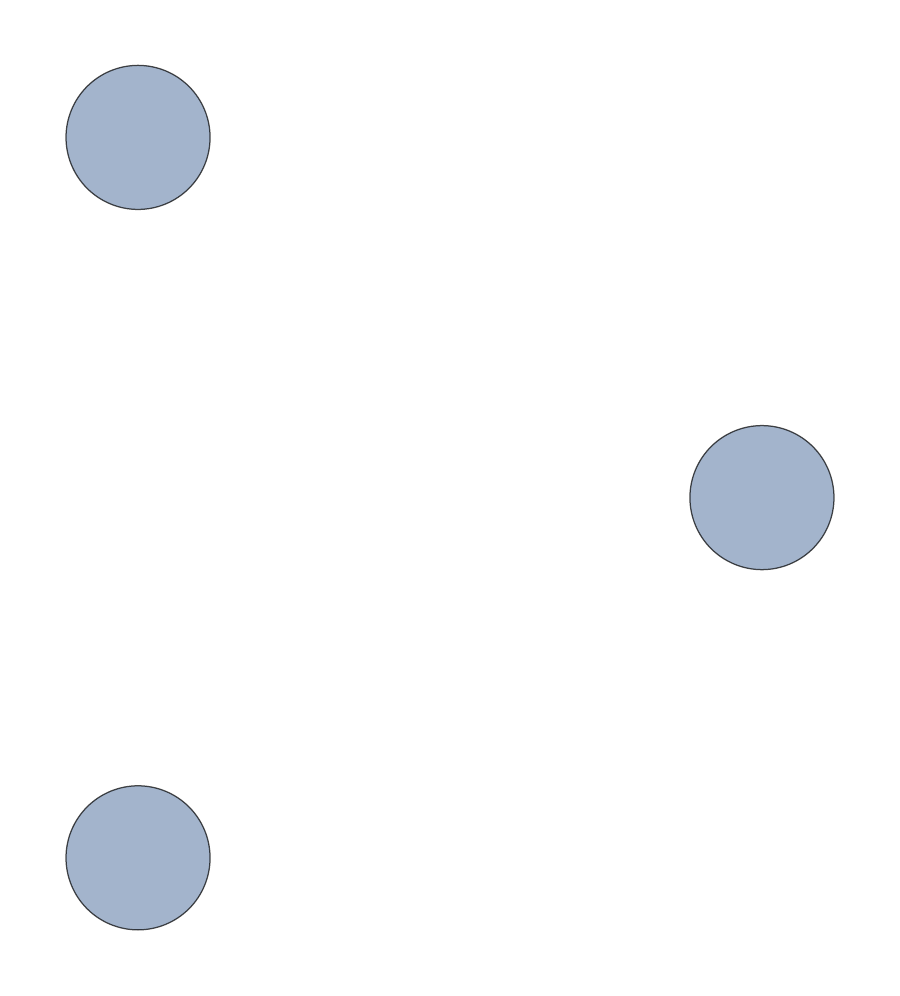}}}\rangle \langle %
\vcenter{\hbox{\includegraphics[width=2em]{pics/K3_subgraph_0_edges_no_label}}}|\;.
\end{align}%
For any value of $p\neq 0,1$ the above two states can be shown by direct calculation to have different ranks, namely rank$(\rho _{G_{1}}^{p}) = 4$, and rank$(\rho _{G_{2}}^{p})=5$.
Therefore, the RG states $\rho_{G_{1}}^{p}$ and $\rho _{G_{2}}^{p}$, defined starting from LU equivalent graph states, cannot even be transformed into each other by a GU operation. In other words, these are not unitary equivalent and, in particular, not LU equivalent. This reasoning can be generalized to an arbitrary number of qubits by introducing the following concepts:

\begin{definition}[$G$-subgraphs state space]
Let $G$ be a graph and $F$ a spanning subgraph of $G$. The space spanned by
the states $|F\rangle $ is called $G$\emph{-subgraphs state space} and is
denoted as%
\begin{equation}
\Sigma _{G}:=\text{ \emph{span}}\left( \{|F\rangle \}_{F\subseteq
G,V_{F}=V_{G}}\right) .
\end{equation}
\end{definition}

This definition prompts to two observations concerned with the complete
graph. The respective proofs are in Appendix \ref%
{sec::proof_of_theorems_in_unitary_equivalence}.

\begin{theorem}[Dimension of $\Sigma _{K_{n}}$]
\label{theorem::dimension_of_complete_subgraphs_state_space}
The $K_{n}$-subgraphs state space $\Sigma _{K_{n}}$ has dimension $2^{n}-n$.
\end{theorem}

\begin{theorem}[Rank of randomized graph states]
\label{theorem::rank_of_random_graph_states}The rank of the randomized graph
state $\rho^p_{K_{n}}$ is $2^{n}-n$, for all $0<p<1$.
\end{theorem}

A direct consequence of Theorem \ref{theorem::rank_of_random_graph_states}
is that the rank of $\rho^p_{K_{n}}$ is maximum over all RG states of $n$
qubits, as long as $p\neq 0,1$. An interesting question is whether there
exists any other randomized graph state $\rho _{G_{n}}$ with maximum rank.
The answer is in the negative. This can be explained by the following argument.
Suppose we have a graph $G_{n}$ given by the complete graph $K_{n}$ where,
without loss of generality, we delete a single edge between vertices $1$ and
$2$. It can be easily seen that the state $|1100...00\rangle $ appears with a
plus sign in the graph state $|G_{n}\rangle $ and all the corresponding
subgraph states. Therefore, the state $|0000...00\rangle -|1100...00\rangle $
cannot be obtained as a superposition of the subgraphs of $G_{n}$ (see the
proof of Theorem \ref{theorem::dimension_of_complete_subgraphs_state_space}
in Appendix \ref{sec::proof_of_theorems_in_unitary_equivalence} for an
explanation). Thus, the rank of $\rho _{G_{n}}$ is always strictly smaller
than $2^{n}-n$.

The above argument also holds for the case of states that correspond to
graphs $G_{n}^{\lnot m}$ with $m$ edges missing with respect to the complete graph, i.e., with $\binom{n}{2}-m$ edges. The rank of the
corresponding RG states is then bounded as
\begin{equation}
\text{rank}(\rho _{G_{n}^{\lnot m}})\leq 2^{n}-n-m.
\end{equation}%
To prove this, the above argument about the state $|1100...00\rangle $
corresponding to $1$'s for the qubits that are not connected by an edge can
be repeated for all the other $m$ pairs of qubits where the edges are
missing, and the above upper bound then follows.
From the above reasoning we can thus infer that the randomized graph state $\rho_{G^{\neg m}_n}$ can never be GU equivalent to $\rho_{K_n}$.
\par
An interesting example in this sense is provided by
the two graph states $\ket{K_n}$ and $\ket{S_n}$, which are known to
be LU equivalent. As we have observed, rank$(K_n) = 2^n-n$, while, since the star graph $S_n$ can be obtained from the complete graph $K_n$ by deleting $\binom{n-1}{2}$ edges, the rank of $\rho_{S_n}$ can be bounded as
\begin{equation}
\text{rank}(\rho _{S_{n}})\leq 2^{n}-n-\binom{n-1}{2}.
\label{eq.::rank_of_RS_graph_state}
\end{equation}%
This proves that, although the star graph state $\ket{S_n}$ and the complete graph state $\ket{K_n}$ are LU equivalent, their corresponding RG states $\rho_{S_n}$ and $\rho_{K_n}$ are not even GU equivalent.

\section{Bipartite Entanglement}

\label{sec::bipartite_entanglement}

In this section, we analyze the bipartite entanglement properties of RG states.
We show that RG states exhibit some properties which are analogous to
bipartite entanglement of pure graph states, while others are different. A
pure graph state is entangled regarding a bipartition if there exists at least
one edge across the partition. The following proposition shows that the same
result holds for RG states.

\begin{proposition}
\label{prop::bipartite_entanglement_of_random_graph_states} Given a graph $G$,
let $A$ and $B$ be disjoint subsets such that $A\cup B=V_{G}$. A RG state
$\rho_{G}^{p}$ is entangled regarding the bipartition $A|B$, if there exists
at least one randomized edge between $A$ and $B$ with randomness $p>0$.
\end{proposition}

\begin{proof}
Let us first consider the graph state composed of two qubits, namely the Bell
state $|\text{Bell}\rangle
=|\vcenter{\hbox{\includegraphics[width=2em]{pics/bell_2}}}\rangle$. The RG
state $\rho_{\text{Bell}}^{p}$ associated to it is thus given by
\begin{equation}
\rho_{\text{Bell}}^{p}=\frac{1}{4}\left(
\begin{array}
[c]{cccc}%
1 & 1 & 1 & 1-2p\\
1 & 1 & 1 & 1-2p\\
1 & 1 & 1 & 1-2p\\
1-2p & 1-2p & 1-2p & 1
\end{array}
\right)  .
\end{equation}
Since the partial transpose of $\rho_{\text{Bell}}^{p}$ has one negative
eigenvalue for $p>0$, $\rho_{\text{Bell}}^{p}$ is entangled whenever $p>0$
\cite{Horodecki1996-11}. Let us now move to the general case and show that
there is always a nonzero probability to project a given RG state $\rho
_{G}^{p}$ onto a randomized Bell state of vertices $a\in A$ and $b\in B$, by
using local $\sigma_{z}$ measurements. Notice that this is never possible if
$\rho_{G}^{p}$ is separable across the bipartition $A|B$. Recall that a
$\sigma_{z}$ measurement on the vertex $v_{i}$ of $|G\rangle$ results in the
graph state $|G-v_{i}\rangle\otimes|+\rangle_{v_{i}}$, where all the edges
touching the vertex $v_{i}$ have been deleted, whenever the outcome $+1 $
occurs \cite{HeinEisertBriegel2004-06}. Therefore, if we now measure all the
vertices except $a$ and $b$, i.e., $V\backslash\{a,b\}$, there is a
nonvanishing probability that all the outcomes are $+1$, and thus a nonzero
probability to delete all the randomized edges of $\rho_{G}^{p}$ except the
one between $a\in A$ and $b\in B$. As a result, there is a nonzero
probability to obtain a randomized Bell state $\rho_{\text{Bell}}^{p}$ between
the vertices $a$ and $b$, which finally shows that the state $\rho_{G}^{p}$ is
entangled with respect to $A|B$ for any $p>0$.
\end{proof}
\par
This shows that, for $p > 0$, RG states show entanglement across any bipartition connected by at least one randomized edge, thus even the action of an imperfect probabilistic CZ gate creates entanglement between the two connected parties.

\bigskip
%%%%% maximally connectedness and persistency
We now consider two different bipartite entanglement properties, namely
maximal connectedness and persistency, specifically introduced in
\cite{BriegelRaussendorf2001-01} for cluster states, and of particular
interest with regard to MBQC. A state is said to be \emph{maximally
connected} if we can project any pair of vertices onto a Bell state with
certainty, by using only local measurements. The following proposition shows
that RG states never enjoy this property.

\begin{proposition}
A randomized graph state is never maximally connected for $p<1$.
\end{proposition}

\begin{proof}
Since for any pair of vertices $\{i,j\}$ there is a nonzero probability
%$(1-p)^{|N_{i}|}+(1-p)^{|N_{j}|}-(1-p)^{|N_{i}|+|N_{j}|}$, \textbf{[WHAT IS Ni?]}
that either vertex $i$ or $j$ is isolated, the state cannot be projected onto
a Bell state $|\text{Bell}\rangle_{i,j}$ with certainty.
\end{proof}

The \emph{persistency} $\mathcal{P}$ of a state is instead the minimal number
of local measurements needed to completely disentangle the state. In Ref.
\cite{BriegelRaussendorf2001-01}, it was shown that, while every cluster state
is maximally connected, the persistency depends on its specific structure.
Results are known for one-dimensional (1D) cluster states $|L_{n}\rangle$, where
the persistency $\mathcal{P}$ equals the Schmidt rank $n/2$, and for
two- or three-dimensional cluster states where $\mathcal{P}$ approaches $n/2$ only
asymptotically. The following proposition shows that the RG state $\rho
_{G}^{p}$ is less robust than the graph state $|G\rangle$.

\begin{proposition}
The persistency of a randomized graph state $\mathcal{P}(\rho_{G}^{p})$ is
always smaller or equal than $\mathcal{P}(|G\rangle)$:
\begin{equation}
\mathcal{P}(\rho_{G}^{p})\leq\mathcal{P}(|G\rangle).
\end{equation}

\end{proposition}

\begin{proof}
Let $\mathcal{P}(|G\rangle)=m$, and $\{M_{1},\cdots,M_{m}\}$ be the
measurements that totally disentangle $|G\rangle$. Then the same set of
measurements $\{M_{1},\cdots,M_{m}\}$ totally disentangles $\rho_{G}^{p}$ too,
as it disentangles each spanning subgraph state of $|G\rangle$. Therefore the
inequality $\mathcal{P}(\rho_{G}^{p})\leq m$ follows.
\end{proof}

The two propositions above show that the bipartite entanglement of a  given RG state is never as robust as the one of the corresponding pure graph state. This observation is expected, due to the method of construction, and is of particular interest with regard to MBQC.

\begin{figure}[b]
\centering
\subfloat[Negativity of all RG states $\rho_{K_n}$ states up to $n=4$ qubits.]{
\includegraphics[width=0.9\linewidth]{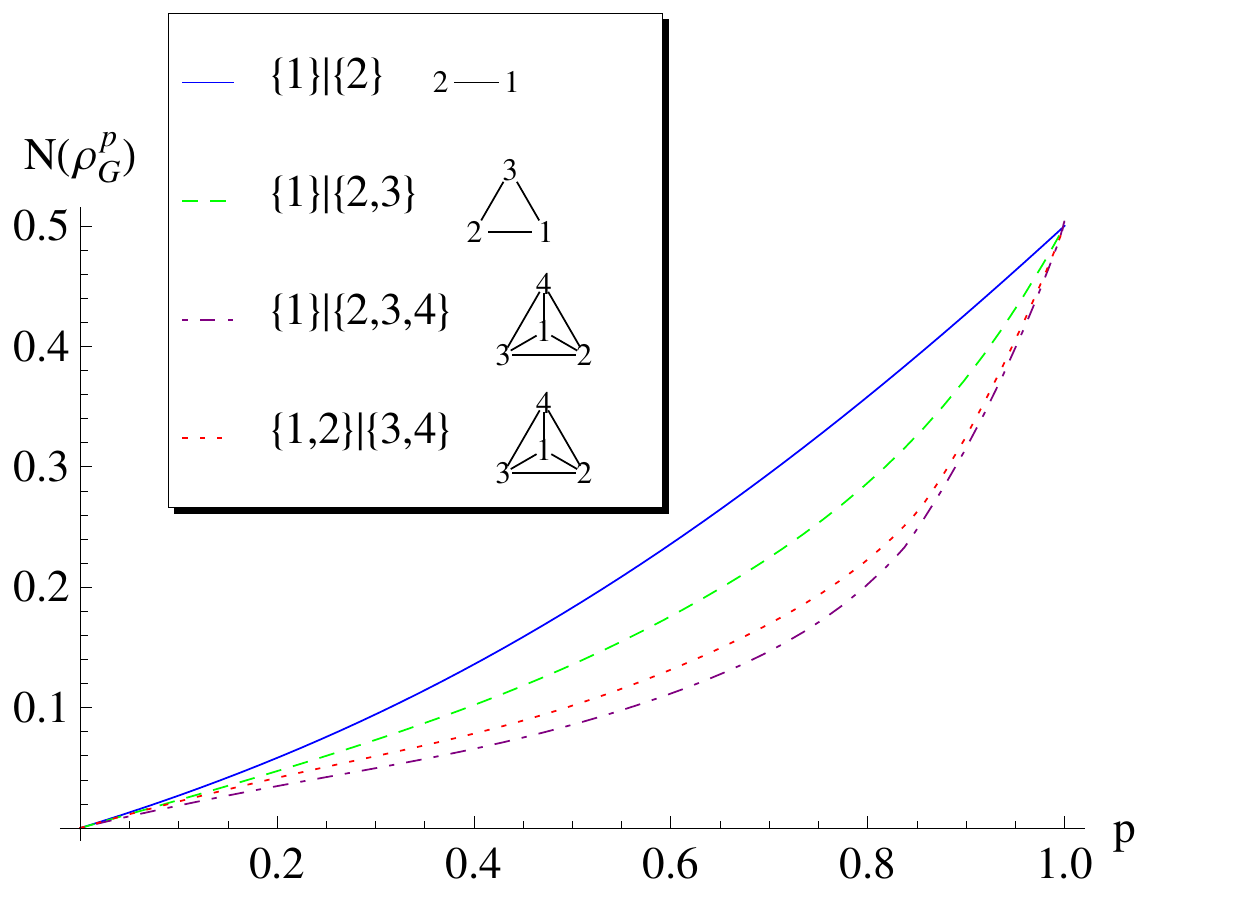}
\label{fig::negativity_of_RK_states}
} \newline%
\subfloat[Negativity of RG states $\rho_{S_n}$ composed of $n=3$ qubits.]{
\includegraphics[width=0.9\linewidth]{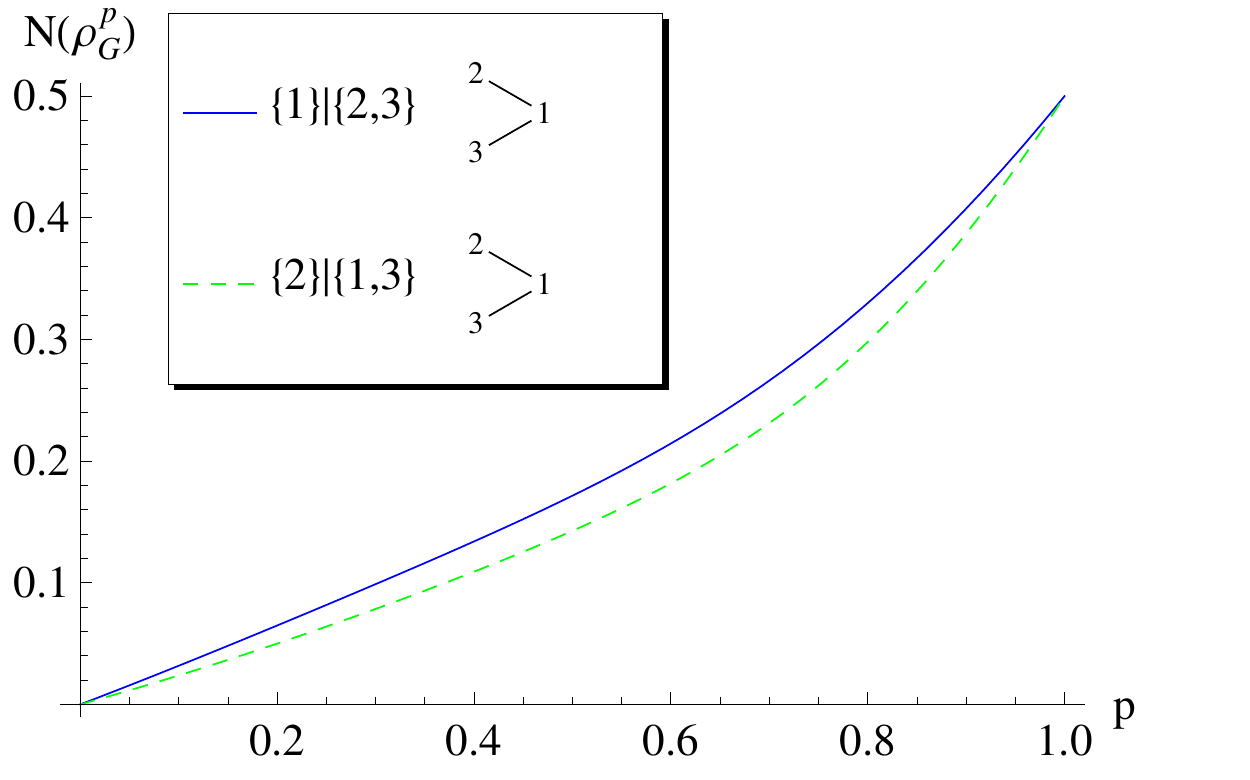}
\label{fig::negativity_of_3v_RG_states}
} \newline%
\subfloat[Negativity of RG states $\rho_{S_n}$ composed of $n=4$ qubits.]{
\includegraphics[width=0.9\linewidth]{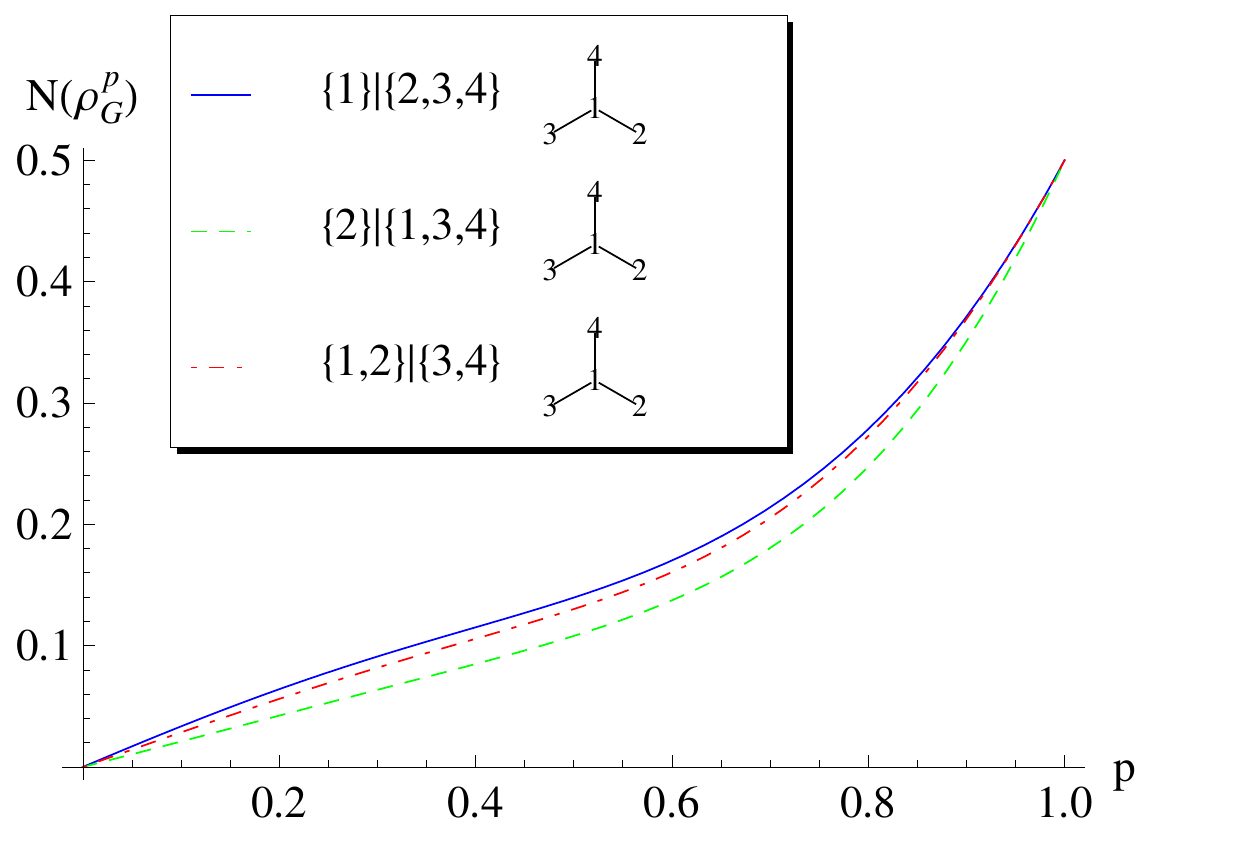}
\label{fig::negativity_of_3v_RG_states}
}
\caption{(Color online) Negativity of some special RG states composed of few
qubits. ``$\{a_{1},\cdots\}|\{b_{1},\cdots\}$" indicates the bipartition with
respect to which the negativity has been
calculated.}%
\label{fig::monotomicity_of_negativity}%
\end{figure}

\par
\bigskip

We finally quantify the amount of bipartite entanglement by considering the
negativity, evaluated with respect to all possible bipartitions of the qubits.
The \emph{negativity} of a bipartite state $\rho_{AB}$ is defined
\cite{VidalWerner2002-02} as
\begin{equation}
N(\rho_{AB})=\frac{||\rho_{AB}^{\Gamma_{A}}||-1}{2},
\end{equation}
where $\Gamma_{A}$ represents the partial transposition with respect to the
subsystem $A$, and $||X||=\text{Tr}[\sqrt{X^{\dagger}X}]$ is the trace norm.
Notice that this is one of the few computable measures of entanglement when
mixed states are concerned.

We have evaluated the negativity numerically for some RG states composed of a
small number of qubits. The results for the negativity of states corresponding
to the complete graph $K_{n}$ and the star graph $S_{n}$ up to $n=4 $ vertices
are reported in Fig. \ref{fig::monotomicity_of_negativity}. As can be seen, in
the studied cases the negativity exhibits a monotonic behaviour in terms of
the randomness parameter $p$. This suggests that the entanglement content
might increase monotonically in $p$ with respect to any bipartition. Actually,
since for the extreme cases $p=0$ and $p=1$ we have a fully separable state
and an entangled state, respectively, one might expect that, as
the weight of entangled subgraph states in $\rho^{p}_{G}$ increases with increasing $p$, a
corresponding growth of the entanglement content of the RG state $\rho^{p}_{G}$.
However, even though this conjecture is supported by numerical
evidence, it is an open question whether the monotonic behavior of the negativity in
terms of the randomness $p$ is a common feature to all RG states.

\section{Genuine Multipartite Entanglement}

\label{sec::GME}

In this section, we consider genuine multipartite entanglement (GME) properties
of RG states. We remind the reader that a state which cannot be written as a
convex combination of biseparable states is called \emph{genuinely
multipartite entangled} (GME) \cite{AcinBLSanpera2001-07}. For
example, in the case of three qubits, a state $\rho$ is genuinely multipartite
entangled, if it can not be expanded in the following decomposition:
\begin{equation}
\rho=c_{1}\rho_{1|23}+c_{2}\rho_{2|13}+c_{3}\rho_{3|12},
\end{equation}
where $\rho_{i|jk}$ is a biseparable state regarding the bipartition
$\{i\}|\{jk\}$, and $\sum_{i=1}^{3}c_i=1$, with $c_i \geq 0$. The condition of being genuine multipartite entangled is thus
stronger than showing bipartite entanglement. As a direct consequence, the
recognition and evaluation of GME becomes much harder, especially for mixed
states. Nonetheless some investigations can be still made for RG states.

As was the case for bipartite entanglement in Fig.
\ref{fig::monotomicity_of_negativity}, we expect the randomness parameter $p$ to tune
the amount of GME of a connected RG state from zero to its maximum value.
Since the two extreme cases $p=0,1$ correspond to a fully separable and a
genuine multipartite entangled state, respectively, we wonder whether the GME
content of a general RG state $\rho_{G}^{p}$ might still follow a
monotonically increasing behavior in terms of $p$.

In order to support this intuition, we have followed the PPT mixer approach
developed in Ref. \cite{JungnitschMGuhne2011-05}. In this approach one uses a
semidefinite program to make an optimization over all fully decomposable witnesses.
An entanglement witness is a Hermitian operator $W$ such that there exists a $\rho$ with $\text{Tr}[W\rho]<0$ and $\text{Tr}[W\rho_{\mathrm{sep}}]\geq0$ for all separable states $\rho_{\mathrm{sep}}$.
A fully decomposable witness $W$ is a witness operator that can be decomposed into two positive semidefinite operators $P_\gamma$ and $Q_\gamma$ for all bipartitions $\gamma$, such that
\begin{equation}
  W=P_\gamma+Q_\gamma^{\Gamma_\gamma},
\end{equation}
with Tr$(W)=1$, $P_{\gamma}\geq0$, $Q_{\gamma}\geq0$ and $\Gamma_{\gamma}$ being the partial transpose regarding bipartition $\gamma$. Such a witness is a GME witness, if there exists a GME state $\rho$ with $\text{Tr}[W\rho]<0$,
and $\text{Tr}[W\rho']\geq0$ for all non-GME states $\rho'$.
With a semi-definite program one can minimize the expectation value Tr($W\rho$) over all fully decomposable witnesses,
such that one can numerically calculate the quantity
\begin{equation}
E_{\text{pptmixer}}(\rho)=\left\vert \min\left(  0,\min_{W\text{ fully decomp.
}}\text{Tr}(W\rho)\right)  \right\vert\label{eq.PPT_mixer_ent_monotone}.%
\end{equation}
Since $E_{\text{pptmixer}}$ is an entanglement monotone, it cannot solely detect the presence of GME
but also bound the amount of GME \cite{JungnitschMGuhne2011-05}. Moreover
it turns out to be necessary and sufficient for entanglement detection
in permutationally invariant states up to three qubits
\cite{NovoMoroderGuhne2013-07}, thus leading to a well defined measure of GME.
Notice that, for graph states and their randomization, only the ones which are
generated by complete graphs are permutationally invariant. Hence we can
solely use this PPT mixer approach as GME measure for the three-qubit RG state
$\rho_{K_{3}}^{p}$, while as a GME monotone for the other RG states. With the
help of the online program \cite{ppt_mixer_online_program}, we obtain the
numerical results for RG states with three, four and five qubits. These are shown
in Fig. \ref{fig::pptmixer_ent_monotone}. The behavior of the monotone of GME
derived from the PPT mixer is monotonic in $p$, supporting our intuition.
Whether the multipartite entanglement of RG states is generally increasing with $p$ remains an open question.
\par

\begin{figure}[h]
\centering \begin{minipage}{\linewidth}
\includegraphics[width=0.9\linewidth]{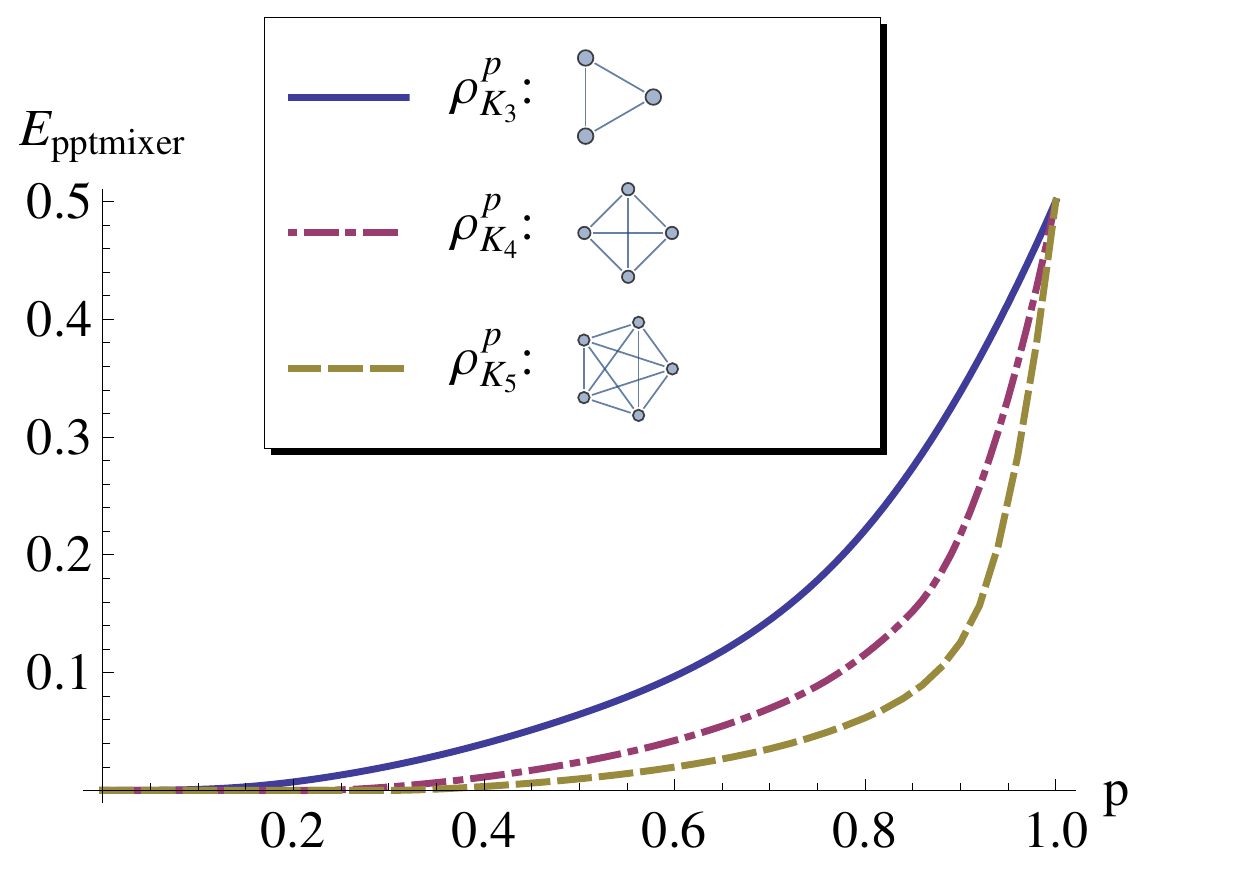}
\end{minipage}
\begin{minipage}{\linewidth}
\includegraphics[width=0.9\linewidth]{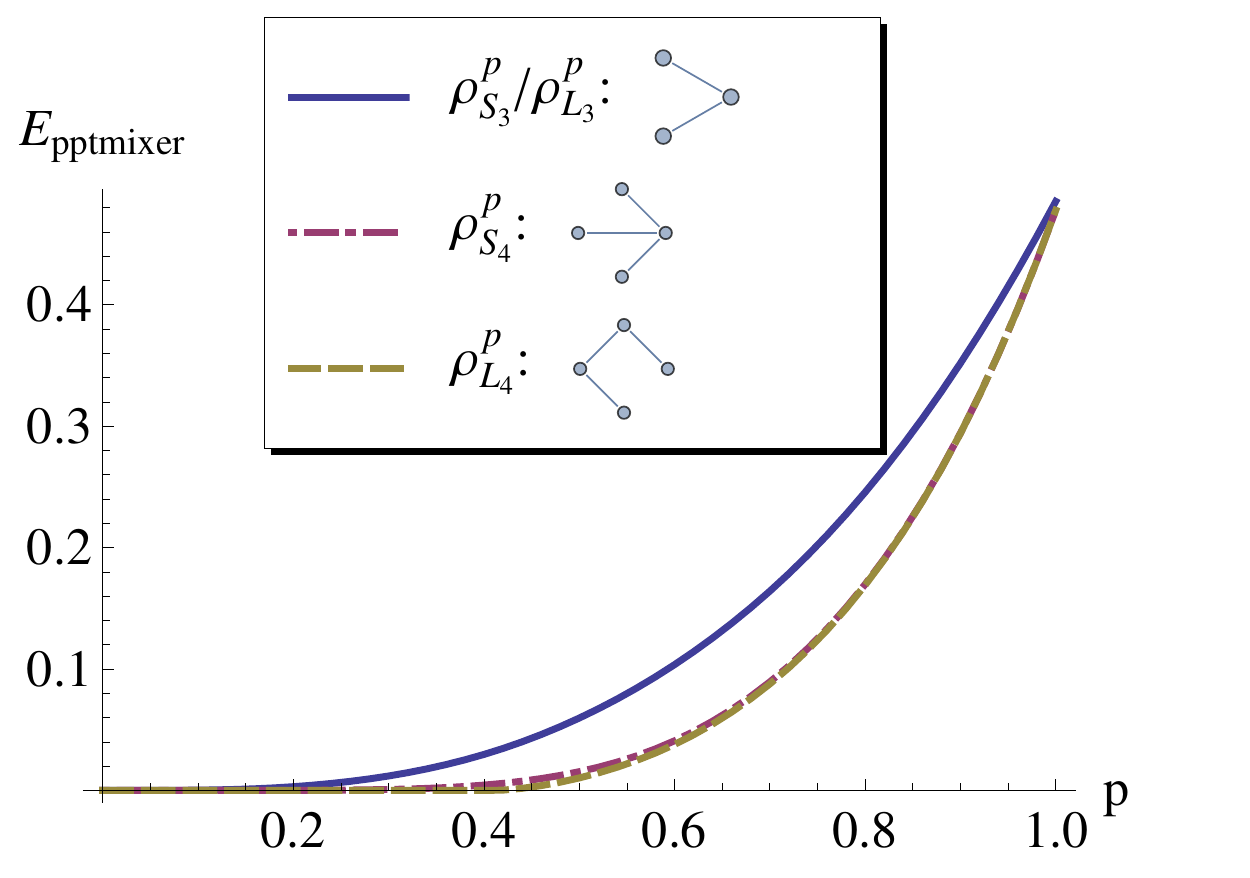}
\end{minipage}
\begin{minipage}{\linewidth}
\includegraphics[width=0.9\linewidth]{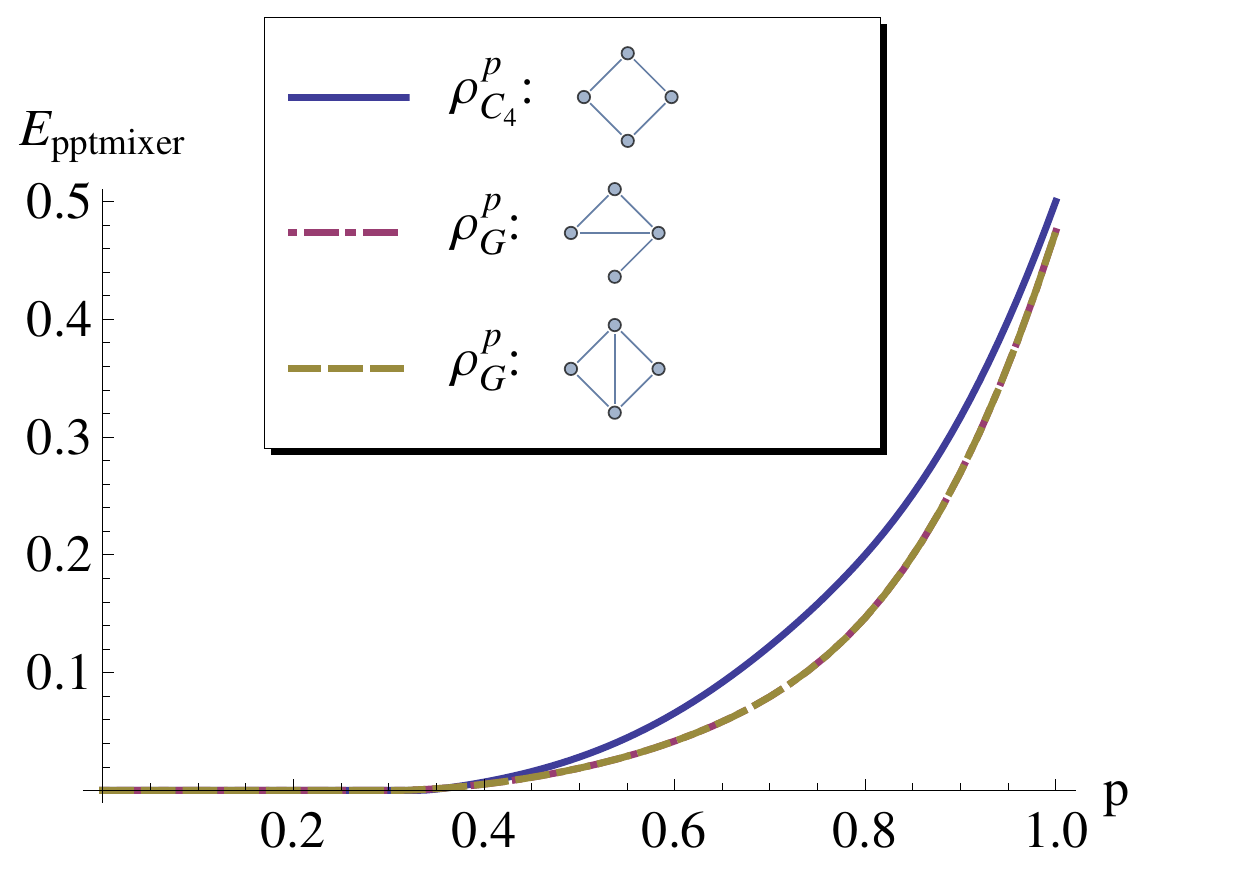}
\end{minipage}
\caption{(Color online) Multipartite entanglement monotone derived from the
PPT mixer as a function of the randomness $p$ for RG states up to five
qubits; see Eq. \eqref{eq.PPT_mixer_ent_monotone}}.%
\label{fig::pptmixer_ent_monotone}%
\end{figure}

\bigskip

If the quantity $\text{Tr}[W\rho_G^p]$ is monotonically decreasing with respect to
$p$, then it allows to us find a critical value of the randomness parameter,
$p_{w}$, such that whenever $p>p_{w}$ the state is guaranteed to show GME. A
depiction of what could happen is illustrated in Fig.
\ref{fig::critical_probability_and_its_upper_bound}. There, the expectation
value of a GME witness on the RG state $\rho_{G}^{p}$ is plotted as a function
of $p$, and compared with the expected behavior of a general measure of GME.
By assuming the existence of a threshold $p_{c}$ above which the state shows
GME (according to the GME measure), it is clear that $p_{w}$ is an upper bound
for $p_{c}$, i.e., $p_{c}\leq p_{w}$. Note that the presence of a threshold $p_c$ is supported by results shown in Fig. \ref{fig::pptmixer_ent_monotone}, and that any negative expectation value for a witness leads to a lower bound for a corresponding entanglement measure \cite{brandao2005quantifying}.

\begin{figure}[t]
\centering
\includegraphics[width=0.9\linewidth]{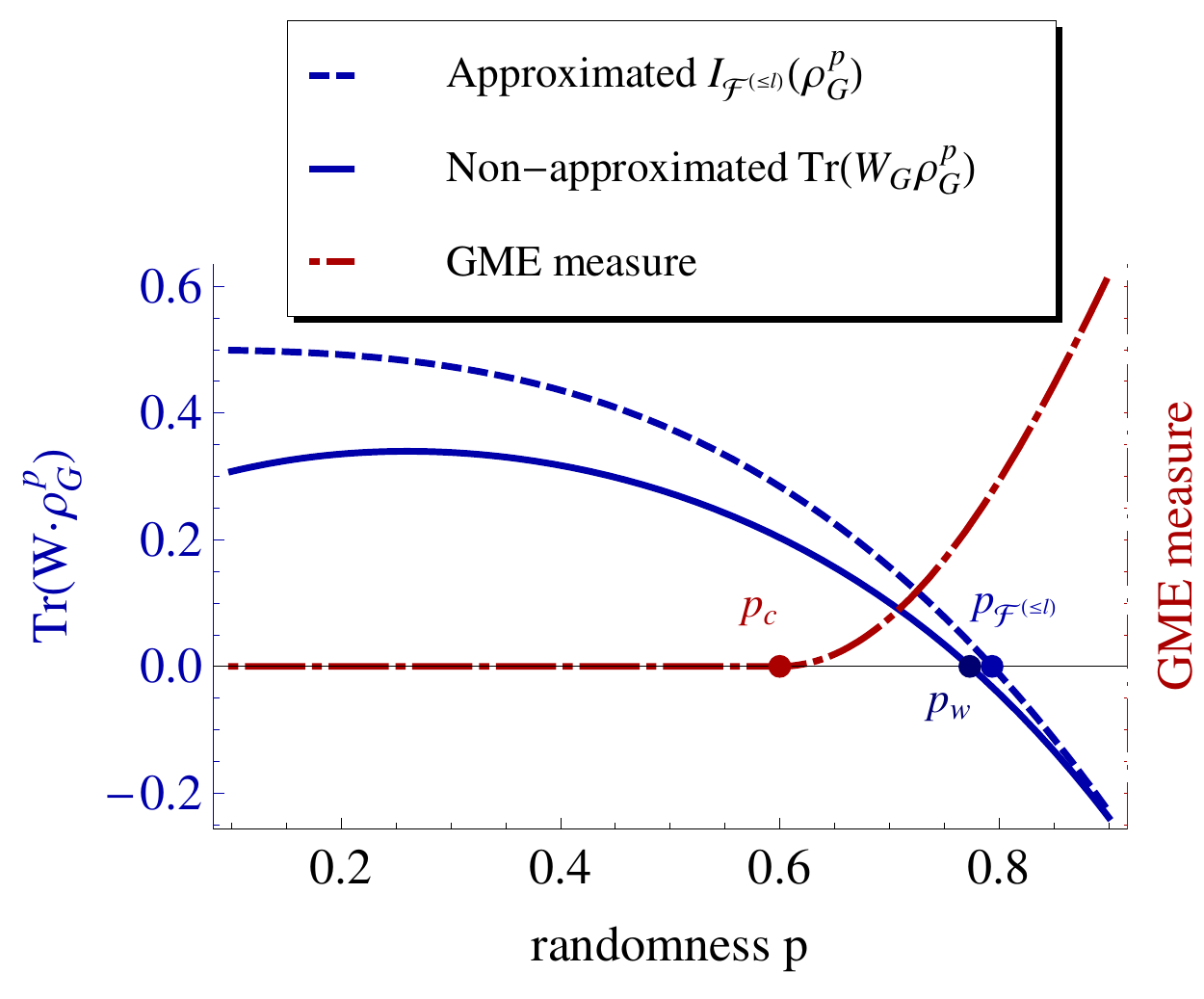}\caption{(Color
online) Relation between a measure of GME and the expectation value of the
witness $W$. The critical probability $p_{c}$ is upper bounded by $p_{w}$,
i.e. the value of $p$ where the expectation value becomes negative. Notice
that the existence of $p_{c}$ and the monotonically increasing behavior of
the GME measure are not guaranteed. The same considerations apply to the
monotonic decreasing behavior of the expectation value. The dashed line
depicts an $l$-level approximated GME witness introduced in Sec.
\ref{sec::upper bound of critical randomness for general graphs}. In contrast
with the non-approximated witness, it is monotonically decreasing for level
$l\leq|E_{G}|/2$ and randomness $1/2\leq p\leq1$. The value of the
non-approximated GME witness $\text{Tr}(W_{G}\rho_{G}^{p})$ is always smaller
or equal than the $l$-approximated GME witness $I_{\mathcal{F}^{(\leq l)}%
}(\rho_{G}^{p})$. }%
\label{fig::critical_probability_and_its_upper_bound}%
\end{figure}

A suitable witness to detect GME in a RG state $\rho_{G}^{p}$ turns out to be
the projector-based witness
\cite{AcinBLSanpera2001-07,GuhneHBELMSanpera2002-12,Bourennane2004-02,TothGuhne2005-02}%
,
\begin{equation}
W_{G}=\frac{1}{2}\mathbbm{1}-|G\rangle\langle
G|.\label{eq.::projector_witness}%
\end{equation}
Notice that the operator above involves only the projector onto the pure graph
state $|G\rangle$ that generates $\rho_{G}^{p}$, disregarding all its
subgraphs. In order to see whether $W_{G}$ of Eq.
\eqref{eq.::projector_witness} provides a negative expectation value for the
state $\rho_{G}^{p}$, one has to compute the overlap $\text{Tr}%
[|G\rangle\langle G|\rho_{G}^{p}]$. Therefore we introduce the next definition:

\begin{definition}
[Randomization overlap]The overlap of a graph state $|G\rangle$ and its
randomization $\rho_{G}^{p}$ is the randomization overlap of $\rho_{G}^{p}$,
\emph{i.e.}
\begin{align}
L(\rho_{G}^{p}):= &  \emph{Tr}[|G\rangle\langle G|\rho_{G}^{p}%
]\label{eq.::randomization_overlap}\\
&  =\sum_{F\emph{\ spans}\text{ }G}p^{|E_{F}|}(1-p)^{|E_{G}\backslash E_{F}%
|}\emph{Tr}[|G\rangle\langle G|F\rangle\langle F|].\nonumber
\end{align}

\end{definition}

Due to the linearity of the trace, the calculation of the randomization
overlap $L(\rho_{G}^{p})$ of Eq. \eqref{eq.::randomization_overlap} thus
reduces to the calculation of the scalar product of the graph state
$|G\rangle$ with all its possible subgraph states $|F\rangle$. Furthermore,
exploiting the symmetric difference defined in section \ref{sec::graphs} and the definition of a graph state in
Eq. \eqref{eq.::def_graphstate}, each contribution $\text{Tr}[|G\rangle\langle
G|F\rangle\langle F|]$ can be rewritten as
\begin{equation}
\text{Tr}[|G\rangle\langle G|F\rangle\langle F|]=\text{Tr}[|G^{\emptyset
}\rangle\langle G^{\emptyset}|G\Delta F\rangle\langle G\Delta F|],
\end{equation}
where $|G^{\emptyset}\rangle$ is associated with the empty graph. Therefore,
the overlap of any two graph states can be recast as the overlap of the graph
defined by the symmetric difference and the empty one. However, even in this
form the scalar product remains highly nontrivial to compute. By the help of
a specifically developed algorithm \cite{XchainsRef}, some special cases can
be computed efficiently and even an analytical formula can be given (see Table
\ref{table::results_of_scalar_product_of_example_graphs}), especially when a
small number of edges is concerned. However, in the general case the overlap
can be given only via some iterative formula \cite{NiekampKGuehne2012}, which
unfortunately scales exponentially in the number of vertices.
\par

\begin{center}
\begin{table}[t]
\centering{\renewcommand{\arraystretch}{1.4}
\begin{tabular}
[c]{|c|c|}\hline
Graph $| G \rangle$ & Overlap $|\langle G^{\emptyset}| G\rangle|^{2}$\\\hline
$L_{2n}$ & $1/2^{2n}$\\\hline
$L_{2n+1}$ & $1/2^{2n}$\\\hline
$C_{2n}$ & $1/2^{2n-2}$\\\hline
$C_{2n+1}$ & $0$\\\hline
$S_{n}$ & $1/4$\\\hline
\end{tabular}
}\caption{Scalar product of some special graph states with $| G^{\emptyset
}\rangle$. The cluster graphs $L_n$ in the table are one-dimensional. The results are
attained by using the formulas derived in Ref. \cite{XchainsRef}.}%
\label{table::results_of_scalar_product_of_example_graphs}%
\end{table}
\end{center}

\bigskip

%%%% another difficulty of the calculation, 2^(E_G:2) edges %%%%%
Besides the difficulty to compute each single overlap, another problem that
inevitably affects the computation of the randomization overlap $L(\rho
_{G}^{p})$ consists of the large number of contributions we have to account
for. As a matter of fact, since a RG state contains $2^{\binom{|E_{G}|}{2}}$
possible subgraphs, that is exponentially increasing in the number of edges,
the number of overlaps contributing to $L(\rho_{G}^{p})$ increases
exponentially fast as well. Nonetheless there exist some special cases that
can be treated explicitly and where an analytical solution can be found. These
cases will be treated in the following, before moving to a possible efficient
approximation of the randomization overlap $L(\rho_{G}^{p})$.

%%%%%%%%%%%%%%%%%%%%%%%%%%%%%%%%%%%%%%%%%%%%%%%%%%%%%%%%%%%%%%%%%%%%%%%%%%%%%%%%%%%%%%%%%%%%%%%%%%%%%%%%%%%%%%
%%  Subsubsection "Upper bound of Critical Randomness of Simple RG States"

\subsection{Calculation of the witness for special RG states}
\label{sec::upper_bound_of_critical_randomness_of_simple_RG_states}

In Appendix \ref{sec::randomization_overlap_of_some_special_RG_states}, we derive
the randomization overlap of both the RG state $\rho_{S_{n}}^{p}$,
corresponding to the star graph $S_{n}$, and the randomized 1D cluster
$\rho_{L_{n}}^{p}$. The expectation value of the witness $W_{S_{n}}$ on the
state $\rho_{S_{n}}^{p}$ takes the form
\begin{equation}
\text{Tr}[W_{S_{n}} \rho_{S_{n}}^{p}]=\frac{1}{4}-\frac{3}{4}p^{n-1},
\end{equation}
which is monotonically decreasing with respect to $p$.
Therefore the threshold probability turns out to be $p_{w}=3^{-1/(n-1)}$, and upper bounds the critical randomness $p_{c}$.

For the randomized 1D cluster state $\rho_{L_{n}}^{p}$ the witness gives
instead the following expectation value:
\begin{align}
\text{Tr} & [W_{L_{n}} \rho_{L_{n}}^{p} ]\nonumber\\
= & \frac{1}{2}-\frac{1}{\sqrt{\lambda_{p}}}\left(  1-\frac{p}{2}+\frac
{\sqrt{\lambda_{p}}}{2}\right)  \left(  \frac{p}{2}+\frac{\sqrt{\lambda_{p}}%
}{2}\right)  ^{n}\nonumber\\
& +\frac{1}{\sqrt{\lambda_{p}}}\left(  1+\frac{p}{2}+\frac{\sqrt{\lambda_{p}}%
}{2}\right)  \left(  \frac{p}{2}-\frac{\sqrt{\lambda_{p}}}{2}\right)
^{n},\label{eq.::witness_1D_RC_state_without_extra_term}%
\end{align}
where $\lambda_{p}=1-p+ p^{2}$ (see Appendix
\ref{sec::randomization_overlap_of_some_special_RG_states} for details).
Notice that this function is also monotonically decreasing with respect to $p$.
Solving the above polynomial in $p$ thus provides an upper bound $p_{w}$ on $p_c$ for
the RG state $\rho_{L_{n}}^{p}$. Both the expectation values above are plotted
in Fig. \ref{fig::pw_of_simple_RG_graphs}. \begin{figure}[t]
\centering
\includegraphics[width=\linewidth]{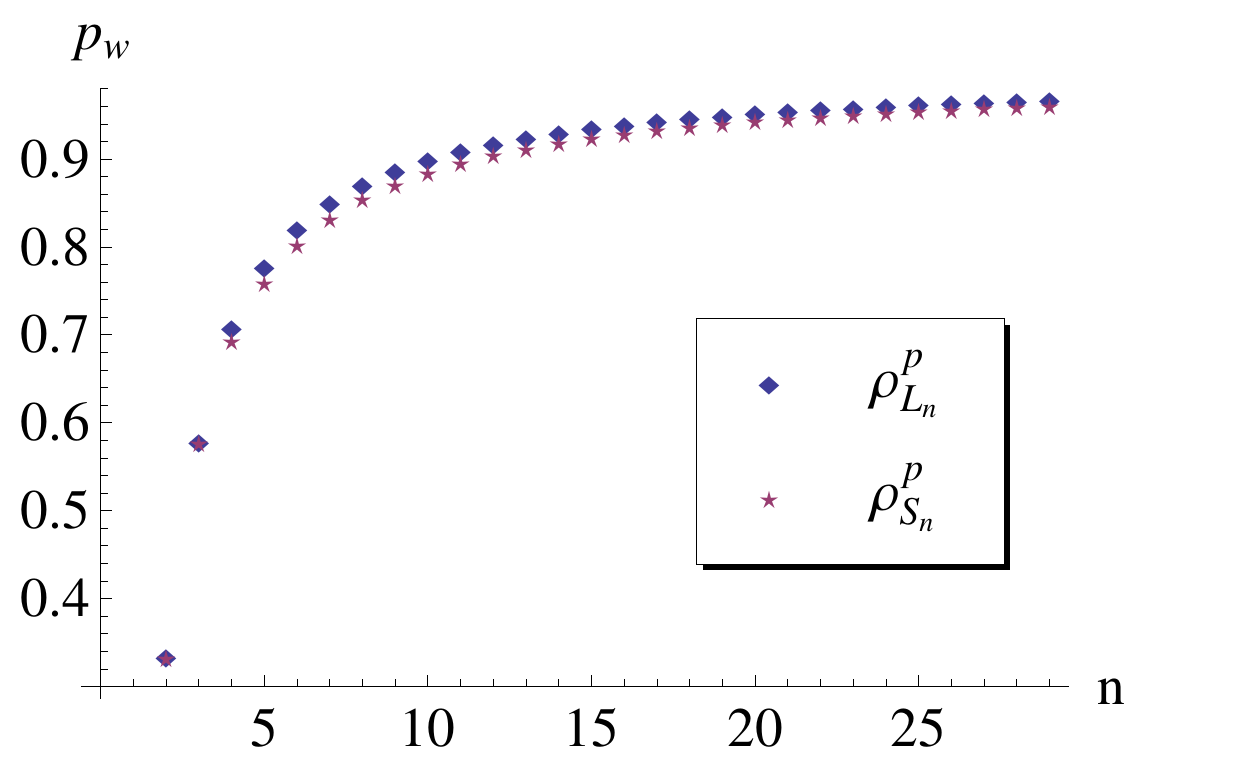}\caption{(Color
online) Probability $p_{w}$ for the randomized star graph state $\rho_{S_{n}%
}^{p}$ and the randomized 1-$d$ cluster state $\rho_{L_{n}}^{p}$.}%
\label{fig::pw_of_simple_RG_graphs}%
\end{figure}

The nonapproximated values $p_{w}$ of the RG cycle state $\rho_{C_{n}}^{p}$
can also be computed numerically by the use of the algorithm developed in Ref.
\cite{XchainsRef}, which will be compared with approximated values in Fig.
\ref{fig::pF_approximation_accuracy} in the next section.

It is worth mentioning that, as expected, $p_{w}$ increases rapidly as the
number of vertices increases. From an experimental point of view, this means
that the more edges one creates, the higher gate quality is required to
guarantee the presence of GME in the final state.
\par
In the following we will follow a different approach, namely we will approximate the witness neglecting all contributions of subgraphs too ``different" from the generating one. This approximation holds whenever the randomness parameter $p$ is high enough.

\subsection{Approximated witness}
\label{sec::upper bound of critical randomness for general graphs}

Due to the structure of a general RG state,
the computation of the scalar product of the pure graph state with all spanning subgraph states turns out to be too complex.
Therefore, we introduce an approximation of the randomization overlap $L(\rho_{G}^{p})$,
that defines the expectation value $W_{G}$. Here we define the $l$\emph{-level
approximation} of a randomization overlap by dropping its subgraph components
$F^{(>l)}$ which differ from $G$ by more than $l$ edges, i.e.,%

\begin{equation}
L_{\mathcal{F}^{(\leq l)}}\left( \rho_{G}^{p}\right) :=\text{Tr}(|G\rangle
\langle G|\rho_{\mathcal{F}^{(\leq l)}}),
\end{equation}
where $\rho_{\mathcal{F}^{(\leq l)}}$ is defined as
\begin{equation}
\rho_{\mathcal{F}^{(\leq l)}}=\sum_{F\text{ s.t. } |E_{F\Delta G}|\leq
l}p^{|E_{F}|}(1-p)^{|E_{G}\backslash E_{F}|}|F\rangle\langle F|.
\end{equation}

The $l$-level approximated witness then reads
\begin{equation}
I_{\mathcal{F}^{(\leq l)}}(\rho_{G}^{p}):=\frac{1}{2}-L_{\mathcal{F}^{(\leq
l)}}\left(  \rho_{G}^{p}\right)  .\label{eq::witness_for_RG_state_low_edges_l}%
\end{equation}
The proof of the next statement is in Appendix
\ref{sec::proof_of_approximation_of_GME_witness}.

\begin{proposition}
\label{prop::monotonicity_of_approximated_witness} The $l$-level approximated
randomization overlap $L_{\mathcal{F}^{(\leq l)}}(\rho_{G}^{p})$ is
monotonically increasing with respect to the randomness $p\geq1/2$ for all
$l\leq|E_{G}|/2$.
\end{proposition}

A good approximation, when $p$ is close enough to $1$, consists in neglecting the subgraphs $\mathcal{F}^{(>2)}$ that differ from $G$ by more than two edges.
This corresponds to a reduced RG state of $|G\rangle$
where only the most relevant subgraphs appear. The following theorem states
that instead of using the full randomization overlap $L(\rho_{G}^{p})$ in the
GME witness, we can focus just on $L_{\mathcal{F}^{(\leq2)}}\left(  \rho
_{G}^{p}\right)  $ with the advantage to make the calculation easier.

\begin{theorem}
[Approximated GME witness]\label{theorem::approximated_GME_witness} Let $G$ be
a graph and $d_{v}$ be the degree of a vertex $v$. The quantity
$L_{\mathcal{F}^{(\leq2)}}\left(  \rho_{G}^{p}\right)  $ is a lower bound for
the randomization overlap $L\left(  \rho_{G}^{p}\right)  $, namely
\begin{align}
L\left(  \rho_{G}^{p}\right)  \geq &  L_{\mathcal{F}^{(\leq2)}}\left(
\rho_{G}^{p}\right)  \label{eq::2-level_approximated_randomization_overlap}\\
= &  p^{\left\vert E_{G}\right\vert }+\frac{1}{4}\left(  1-p\right)
p^{\left\vert E_{G}\right\vert -1}\left\vert E_{G}\right\vert \nonumber\\
&  +\frac{1}{2^{4}}\left(  1-p\right)  ^{2}p^{\left\vert E_{G}\right\vert
-2}\left[  \binom{\left\vert E_{G}\right\vert }{2}+3\sum_{v\in V_{G}}%
\binom{d_{v}}{2}\right]  .\nonumber
\end{align}
For $p\gg1/2$, $L\left(  \rho_{G}^{p}\right)  \simeq L_{\mathcal{F}^{(\leq2)}%
}\left(  \rho_{G}^{p}\right)  $. The following quantity can be regarded as
a GME witness for $\rho_{G}^{p}$.
\begin{equation}
I_{\mathcal{F}^{(\leq2)}}(\rho_{G}^{p}):=\frac{1}{2}-L_{\mathcal{F}^{(\leq2)}%
}\left(  \rho_{G}^{p}\right)  \label{eq::witness_for_RG_state_low_edges_2}%
\end{equation}
If $I_{\mathcal{F}^{(\leq2)}}\left(  \rho_{G}^{p}\right)  < 0$, it is then
guaranteed that the RG state $\rho_{G}^{p}$ is genuinely multipartite entangled.
\end{theorem}

See Appendix \ref{sec::proof_of_approximation_of_GME_witness} for a proof.
Notice that the value of the randomness parameter $p_{\mathcal{F}}$ that makes
$I_{\mathcal{F}^{(\leq2)}}(\rho_{G}^{p})$ vanishing is still an upper bound of
the critical randomness $p_{c}$ for the RG state $\rho_{G}^{p}$. Notice that
by construction the following chain of inequalities holds $p_{c}\leq p_{w}\leq
p_{\mathcal{F}}$. Furthermore, according to Proposition \ref{prop::monotonicity_of_approximated_witness}, the witness $I_{\mathcal{F}^{(\leq2)}}(\rho_{G}^{p})$ is monotonically decreasing as a
function of $p$. Hence whenever $p>p_{\mathcal{F}}$ the RG state $\rho
_{G}^{p}$ shows GME.

\begin{figure}[t]
\centering \begin{minipage}{\linewidth}
\includegraphics[width=0.8\columnwidth]{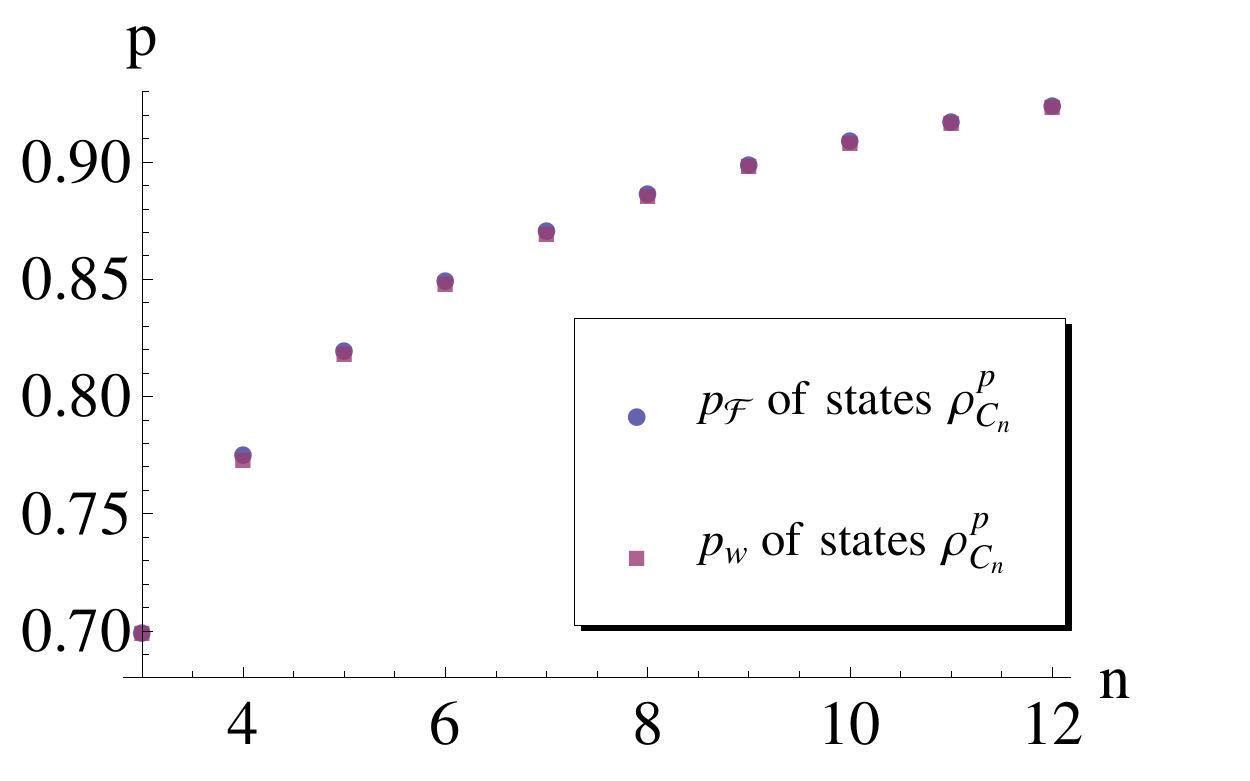}\\
\end{minipage}
\begin{minipage}{\linewidth}
\includegraphics[width=0.8\columnwidth]{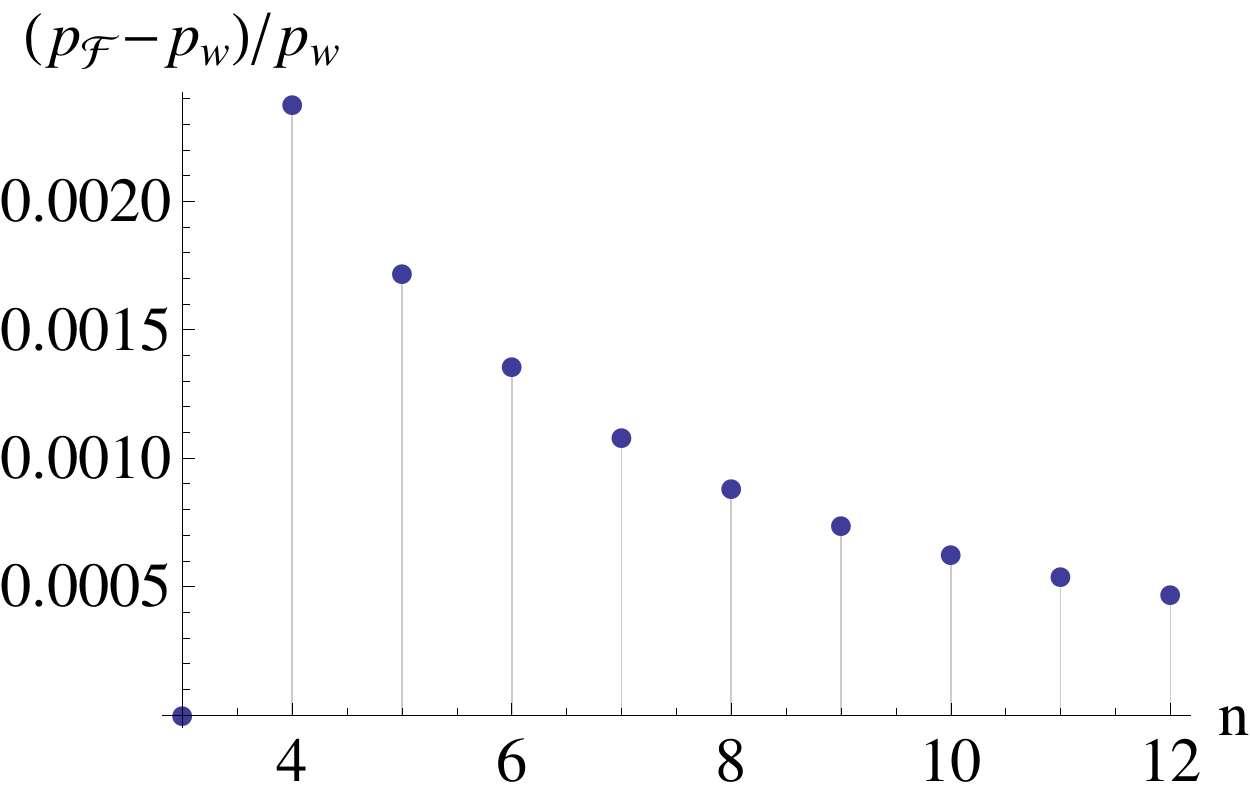}\\
\end{minipage}
\caption{(Color online) Accuracy of the approximated GME witness
$L_{\mathcal{F}^{(\leq2)}}( \rho_{C_{n}}^{p}) $ for the cycle RG graph
$\rho_{C_{n}}^{p}$. The parameter for comparison is the threshold probability
$p_{w}$, calculated according to algorithm explained in Ref. \cite{XchainsRef}%
.}%
\label{fig::pF_approximation_accuracy}%
\end{figure}

By employing this theorem one can detect GME even for a graph with relatively
many edges, however a study about how well the approximated witness performs
is now needed. In order to check the accuracy of our approximation, we
consider as an example the cycle RG graph $\rho_{C_{n}}^{p}$ and plot the
relative difference between $p_{\mathcal{F}}$ and $p_{w}$. As we can see in
Fig. \ref{fig::pF_approximation_accuracy}, for $n=3$ the value of
$p_{\mathcal{F}}$ equals $p_{w}$, while for higher $n$ the approximation
becomes more and more accurate as the number of vertices increases. Note that
the equality for $n=3$ results from the fact that the single neglected
contribution $\mathrm{Tr}[|G^{\emptyset}\rangle\langle G^{\emptyset}%
|C_{3}\rangle\langle C_{3}|]$ in $L_{\mathcal{F}^{(\leq2)}}(\rho_{C_{3}}^{p})$
is equal to zero.

\begin{figure}[h]
\centering
\includegraphics[width=\columnwidth]{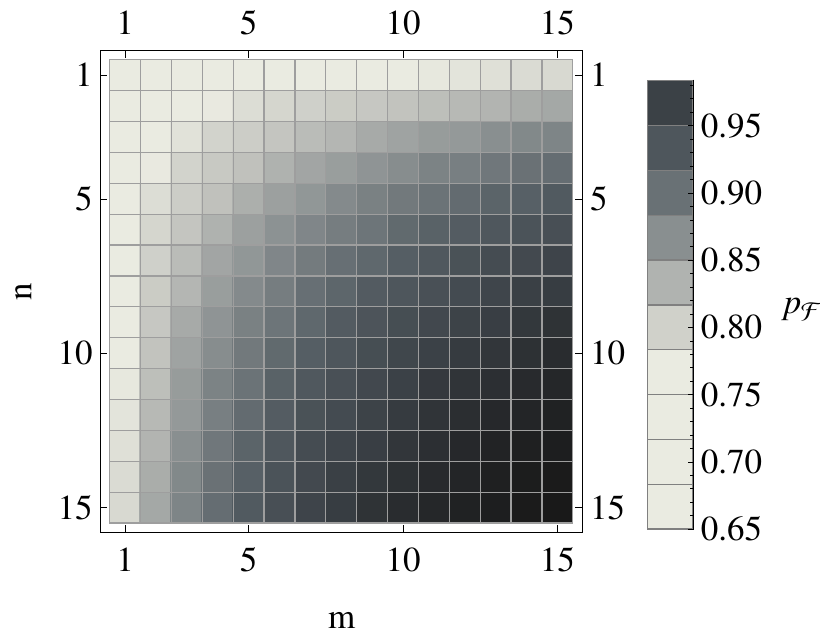}\caption{Threshold
probability $p_{\mathcal{F}}$ for randomized 2D cluster states
$\rho_{L_{{n\times m}}^{p}}=R_{p}(|L_{n\times m}\rangle)$. Here, $m$ and $n$
represent the number of vertices along the $x$ and $y$ axes of the 2D
cluster, respectively. The quantity $p_{\mathcal{F}}$ is depicted as a map in
a $(m,n)$ grid.}%
\label{fig::pF_2d_cl_2edges}%
\end{figure}

\begin{figure}[h]
\centering
\includegraphics[width=.9\columnwidth]{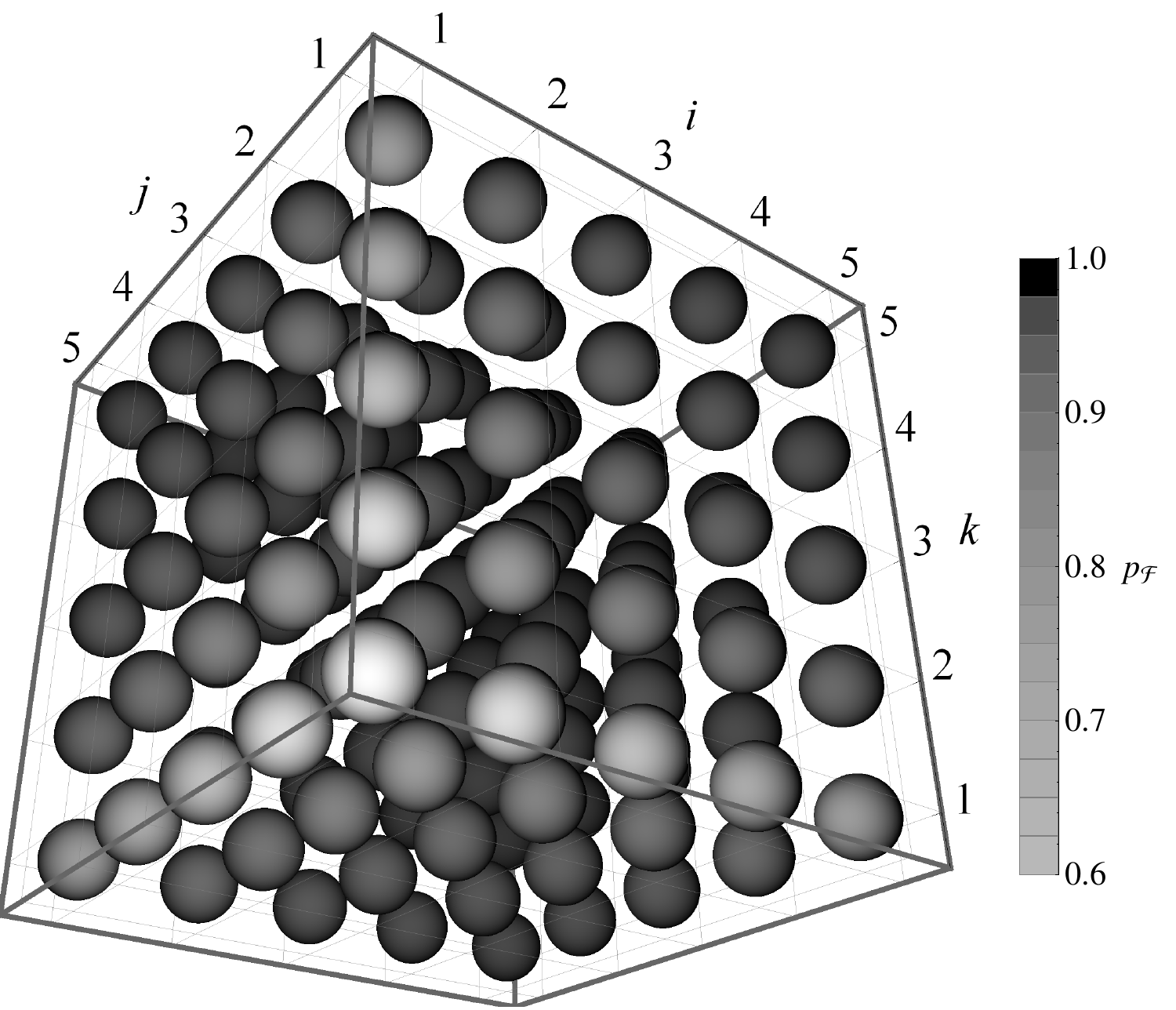}\caption{Threshold
probability $p_{\mathcal{F}}$ for randomized 3D cluster states
$\rho_{L_{i\times j\times k}}^{p}=R_{p}(|L_{i\times j \times k}\rangle)$. The
indices $i$, $j$ and $k$ represent the number of vertices along the $x$, $y$,
and $z$ axes of the 3D cluster, respectively. The quantity $p_{\mathcal{F}%
}$ is depicted in grayscale in a $(i,j,k)$ grid.}%
\label{fig::pF_3d_cl_2edges}%
\end{figure}

In order to show the quality of our approximation we consider here other
relevant RG states, that is randomized 2D and 3D cluster states. For
these states we plot the approximated $p_{\mathcal{F}}$ in Figs.
\ref{fig::pF_2d_cl_2edges} and \ref{fig::pF_3d_cl_2edges}, as a function of
the number of vertices along each direction of the cluster. As we can see in
Fig. \ref{fig::pF_2d_cl_2edges}, $p_{\mathcal{F}}$ for the two-dimensional RG
state $\rho_{L_{m\times n}}^{p}$ increases as the sum $m+n$ grows, where $m$
and $n$ are the number of vertices along the $x$ and $y$ axes, respectively.
It also turns out that the values of $p_{\mathcal{F}}$ for two RG cluster
states $\rho_{L_{m_{1}\times n_{1}}}^{p}$ and $\rho_{L_{m_{2}\times n_{2}}%
}^{p}$ are very close to each other whenever $m_{1}+n_{1}=m_{2}+n_{2}$ . The
same arguments hold also for the three-dimensional randomized cluster state
(see Fig. \ref{fig::pF_3d_cl_2edges}).

\par

Notice that the approximated witness given in Eq. \eqref{eq::witness_for_RG_state_low_edges_2} can be exploited to obtain a value of the randomness parameter $p$ above which the RG state shows GME. Vice versa, if we have at disposal only CZ gates with a fixed parameter $p$, we can then use the estimates given by the witness to find out possible multipartite entangled RG states one could create (see Figs. \ref{fig::pw_of_simple_RG_graphs} and \ref{fig::pF_approximation_accuracy}).

\section{Bell inequalities}
\label{sec:bell}

In this section, we investigate when RG states cannot be
described in terms of \emph{local hidden variable} (LHV) models
\cite{bell1964einstein, mermin1990extreme, greenberger1990bell}. Any LHV model
has to fulfill the constraints of realism and locality. These two facts result
in bounds on the strength of correlations, which can be formally captured in
terms of Bell inequalities\cite{bell1964einstein}. A violation of such an
inequality excludes the description of the correlations in terms of an LHV
model \cite{mermin1990extreme, greenberger1990bell}. We will show that RG
states violate Bell inequalities developed for pure graph states, whenever the
randomization parameter $p$ is high enough. In order to do so we review the
stabilizer description of graph states \cite{hein2006entanglement}.

Given a graph $G$, we can associate to each vertex $i$ a
\emph{stabilizing operator} $g_{i}$ as follows
\begin{equation}
g_{i}=X^{(i)}\bigotimes_{j\in N(i)}Z^{(j)}.
\end{equation}
where $N(i)$ is the neighborhood of the vertex $i$, i.e., the set of
vertices connected to $i$. Here, $X^{(i)},Y^{(i)},Z^{(i)}$ denote the Pauli
matrices $\sigma_{x},\sigma_{y},\sigma_{z},$ acting on the $i$-th qubit. The
graph state $|G\rangle$ associated with the graph $G$ is the unique $n$-qubit
state fulfilling
\begin{equation}
g_{i}|G\rangle=|G\rangle,\mbox{ for }i=1,...,n.
\end{equation}
The $n$ operators $g_{i}$ turn out to be the generators of a group, called
stabilizer and denoted by $S(G)$. The group $S(G)$ can be shown to be Abelian
and is composed of $2^{n}$ elements $s_{j}$. By this definition it
straightforwardly follows that $\langle G|s_{j}|G\rangle=1$ for any
$j=1,...,2^{n}$. As any $s_{j}$ can be expressed as a product of $n$
dichotomic local observables, we can thus define the following Bell operator
\cite{guhne2005bell}:
\begin{equation}
\mathcal{B}(G)=\frac{1}{2^{n}}\sum_{j=1}^{2^{n}}s_{j}.
\end{equation}
Furthermore since a graph state is a product of projectors of
its stabilizer generators, i.e., $|G\rangle\langle G|=\prod
_{i}(\mathbbm{1}+g_{i})/2=\mathcal{B}(G)$, the expectation value of
$\langle\mathcal{B}(G)\rangle$ reaches its maximum value $1$ only for the
state $|G\rangle$. By defining the quantity
\begin{equation}
\mathcal{D}(G)=\max_{\text{LHV}}|\langle\mathcal{B}(G)\rangle|,
\end{equation}
where the maximum is taken over all LHV models, equivalently
taken over all possible expectation values of local observables $\langle
X^{(i)}\rangle$, $\langle Y^{(i)}\rangle$, $\langle Z^{(i)}\rangle$ within
$\{-1,+1\}$, we then have the following Bell inequality \cite{guhne2005bell}
\begin{equation}
\langle\mathcal{B}(G)\rangle\leq\mathcal{D}(G).\label{bell_in}%
\end{equation}
As a straightforward consequence, given the graph state $|G\rangle$, we are
guaranteed that it cannot be described by a LHV model whenever $\mathcal{D}(G)<1$.

For our purpose it is more convenient to rephrase the Bell inequality
\eqref{bell_in} in terms of a detection operator. Keeping in mind that the
Bell operator $\mathcal{B}(G)$ is exactly the projector $|G\rangle\langle G|
$, the following witness operator can be found \cite{guhne2005bell,
HyllusGBLewenstein2005-07}
\begin{equation}
W_{\text{LHV}}=\mathcal{D}(G)\mathbbm{1}-|G\rangle\langle G|.
\end{equation}
Hence, whenever $\Tr[W_{\text{LHV}}\rho]<0$, i.e., the expectation value
of $W_{\text{LHV}}$ on the quantum state $\rho$ is negative, the state $\rho$
violates local realism, and thus cannot be described by LHV models.
Note that the witness $W_{\text{LHV}}$ is similar to the witness
for GME of Eq. \eqref{eq.::projector_witness}. They indeed differ only in the
value of the coefficient of the identity operator. Notice furthermore that the
approximation techniques developed so far apply here too, allowing us to
proceed as in Eq. \eqref{eq::witness_for_RG_state_low_edges_l} in the previous
section, i.e.,%

\begin{equation}
I^{(\leq l)}_{\mathrm{LHV}}(\rho_{G}^{p}):=\mathcal{D}(G)-L_{\mathcal{F}%
^{(\leq l)}}\left(  \rho_{G}^{p}\right)
.\label{eq::Bell_ineq_for_RG_state_low_edges_l}%
\end{equation}

In \cite{guhne2005bell}, the quantity $\mathcal{D}(G)$ has been
calculated for different graphs with number of qubits $n$ up to $10$. Our
analysis consists of calculating the approximated threshold $p_{\text{LHV}%
}^{l\leq2}$ for a given graph state $|G\rangle$, such that $I_{\mathrm{LHV}%
}^{(\leq2)}(\rho_{G}^{p_{\text{LHV}}})=0$. Since $I_{\mathrm{LHV}}^{(\leq
2)}(\rho_{G}^{p})$ is monotonically decreasing with respect to $p$ for $p>1/2$
(see Proposition \ref{prop::monotonicity_of_approximated_witness}), any
randomness parameter $p>p_{\text{LHV}}^{(l\leq2)}$ will then lead to a RG
state that cannot be described in terms of a LHV model.

In Fig. \ref{fig::p_LHV_Bell_ineq}, we show the achieved result for several important
RG states. In this figure one can see that the classical bounds $\mathcal{D}(G)$ are crucial for the behavior of $p_{\text{LHV}}$. For a given type of graph, since the classical bound $\mathcal{D}(G)$ is decreasing with respect to the number of vertices $n$, the threshold $p_{\text{LHV}}$ is not monotonically increasing with respect to $n$. The ordering of $p_{\text{LHV}}$ among different types of graphs can be explained via the ordering of $\mathcal{D}(G)$. For $n\leq5$, $\mathcal{D}(C_n)=\mathcal{D}(L_n)=\mathcal{D}(S_n)$ holds. Therefore $p_{\text{LHV}}(C_n)>p_{\text{LHV}}(L_n)>p_{\text{LHV}}(S_n)$ has the same ordering as the threshold $p_{\text{GME}}$ for GME; see Figs. \ref{fig::pw_of_simple_RG_graphs} and \ref{fig::pF_approximation_accuracy}. For $n>5$, the ordering of the threshold values $p_{\text{LHV}}(S_n)>p_{\text{LHV}}(L_n)>p_{\text{LHV}}(C_n)$ reflects the ordering of the classical bounds for the different types of graphs, i.e., $\mathcal{D}(S_n)>\mathcal{D}(L_n)>\mathcal{D}(C_n)$. For larger $n$, we observe that the nonlocality of the randomized star graph states is fragile with respect to our noise model. This is analogous to the noise resistance of GME for star graph states. The fragility of GME states for other noise models has been investigated in \cite{AliGuhne2014-02}.

\par

Similar to the previous section, we can use the results provided by $I_{LHV}$ of Eq. \eqref{eq::Bell_ineq_for_RG_state_low_edges_l} in order to generate nonlocal multiqubit states by using only CZ gates with a given success probability $p$. For instance, if we have CZ gates with success probability $p=0.84$, we can then create a nonlocal six-qubit system via generating a six-qubit randomized cycle graph state by subsequently connecting the six qubits using solely the CZ gates at disposal.

\begin{figure}[ptb]
\centering
\includegraphics[width=0.9\linewidth]{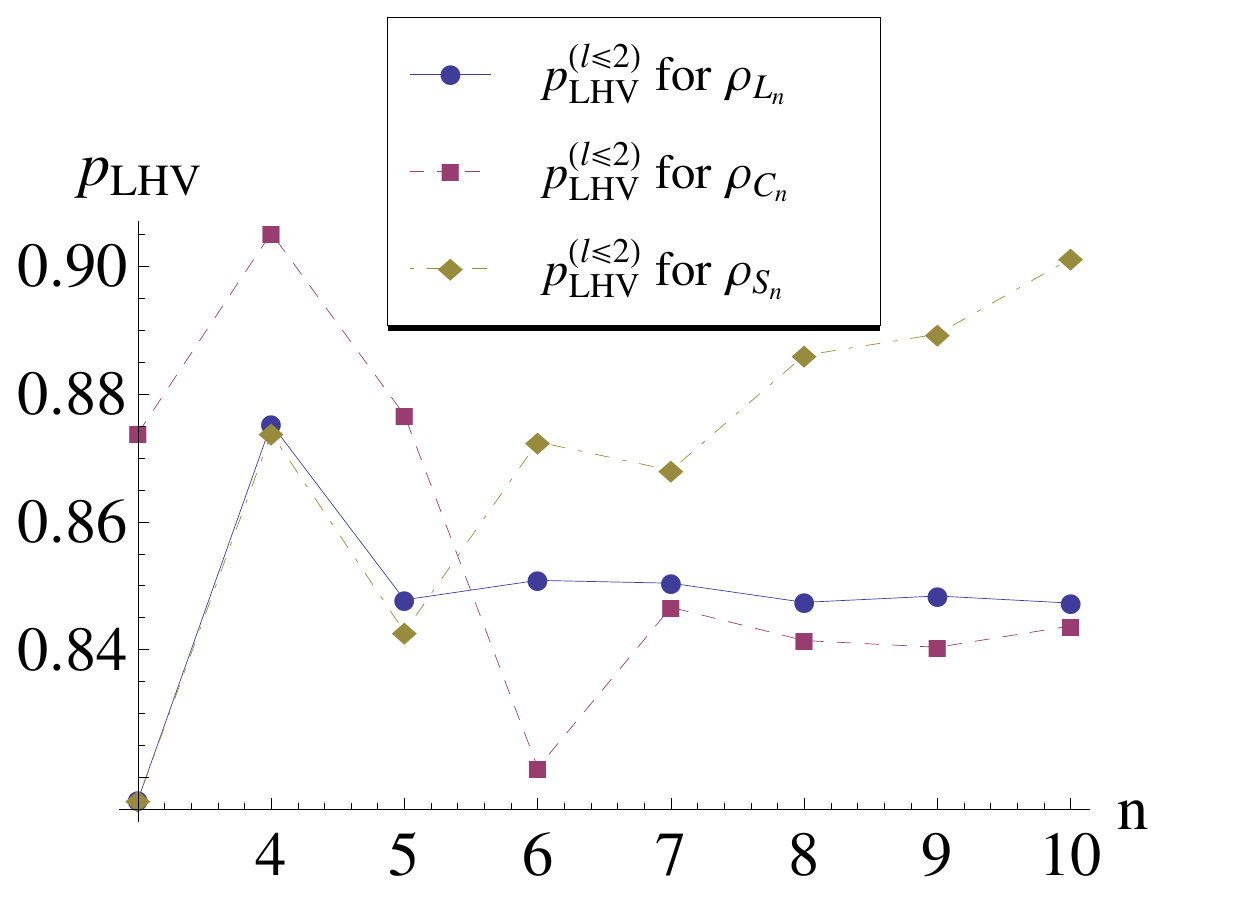}\caption{(Color
online) The probability thresholds $p^{(l \leq2)}_{\mathrm{LHV}%
}$ for some important RG states. These thresholds are the zero-crossings of
Eq. \eqref{eq::Bell_ineq_for_RG_state_low_edges_l}. Due to the complexity of the calculation of the classical bounds $\mathcal{D}(G)$, only the thresholds for the states up to $10$ qubits are analyzed . The behavior of $p_{\text{LHV}}^{(l\leq2)}$ is explained at the end of the section \ref{sec:bell}. } %
\label{fig::p_LHV_Bell_ineq}%
\end{figure}

\section{Conclusions}

\label{sec::conclusion} %Topic
In this paper, we introduced a class of $n$-qubit mixed states that we
called randomized graph (RG) states because they can be derived from pure
graph states by applying a randomization procedure. They represent a quantum
analog of random graphs. These states can also be regarded as the
resulting states in an imperfect graph state generation procedure
\cite{lim2005repeat, lim2006repeat,beige2007repeat}. We studied
in particular the entanglement properties of such states and it turned out
that their entanglement classification is quite different from the one for
graph states.
We investigated whether local unitary (LU)equivalence of pure graph states implies
LU equivalence of their randomized version, and answered this question in
a negative way. Although the presence of a randomized edge
guarantees bipartite entanglement between the two parties that are linked by
the edge, the bipartite entanglement of RG states is more fragile under the
action of local measurements with respect to the one of their corresponding
graph states. We investigated this aspect by evaluating the connectedness
and persistency of RG states.
%This is exhibited by their maximally connectedness and persistency.
We then studied the multipartite entanglement properties of RG states. Due
to the fact that these multi-qubit states are mixed, we could evaluate the
multipartite entanglement content only in some particular cases, namely for
states up to four qubits. In such cases we
could show that multipartite entanglement exhibits a monotonic behavior as
a function of the randomness parameter $p$, while it is still an open
problem whether the entanglement of a general RG state grows monotonically
with $p$. In the general case we could define a critical value $p_c$ for the
randomness parameter above which the RG states are guaranteed to be
multipartite entangled by employing suitable multipartite entanglement
witnesses.
The threshold $p_c$ also provides an estimate of how much noise the CZ gates can be in order to guarantee GME in the generated state.
Furthermore, the same approach was exploited
to study the possibility to describe such RG states in terms of local hidden variable (LHV) models.
Again, we could find a critical probability $p_\text{LHV}$ above which the
quantum state surely violates a Bell inequality.
The threshold $p_{LHV}$ also gives a hint regarding which kind of nonlocal multi-qubit states can be created by using solely controlled-Z gates with a given success probability.

%possible applications
We point out that RG states have possible applications in measurement based
quantum computation, quantum key distribution, quantum networks, etc.
Since RG states are derived by the use of imperfect controlled-Z gates, which is unavoidable in a laboratory, it is more natural to consider these states instead of pure graph states in the quantum information processing task one wants to pursue.

\par
As an outlook, the emergence of giant components of RG states and the properties of RG states in the asymptotic limit $n\rightarrow\infty$ are interesting theoretical topics that deserve further investigation.
Other interesting questions that still need to be addressed are for example
the possibility of identifying a Hamiltonian which has a RG state as
eigenstate, or the possibility of designing a protocol to herald the
components of a RG state, such that one can perform a preselection of the RG
state to extract certain subgraph states from it.

\begin{acknowledgments}
  We thank O.G\"uhne for useful discussions. J.W., D.B. and H.K. were financially supported by DFG and SFF of Heinrich-Heine-University D\"usseldorf. K.L.C. was supported by the National Research Foundation \& Ministry of Education, Singapore. S.S. was supported by the Royal Society.
\end{acknowledgments}

%%%%%%%%%%%%%%%%%%%%%%%%%%%%%%%%%%%%%%%%%%%%%%%%%%%%%%%%%%%%%%%%%%%%%%%%%%%%%%%%%%%%%%%%%%%%%%%%%%%%%%%%%%%%%%
%%   Appendix
%
\appendix

\section{Proofs of theorems \protect\ref%
{theorem::dimension_of_complete_subgraphs_state_space} and \protect\ref%
{theorem::rank_of_random_graph_states}}

\label{sec::proof_of_theorems_in_unitary_equivalence} The proofs of Theorems %
\ref{theorem::dimension_of_complete_subgraphs_state_space} and \ref%
{theorem::rank_of_random_graph_states} are given below. Notice that for
Theorem \ref{theorem::dimension_of_complete_subgraphs_state_space} two
proofs are provided, the former being more intuitive, the latter being more
formal.

\begin{proof}[Proof of Theorem \protect\ref%
{theorem::dimension_of_complete_subgraphs_state_space}]
Let us denote the $n$-qubit state with a single qubit in state 1 at position
$i$ as $|1_{i}\rangle $. Then, from the definition of graph states in terms
of CZ operations (Eq. \eqref{eq.::def_graphstate}) it follows that the $n$
linearly independent (but not mutually orthogonal) states given by
\begin{equation}
|00...0\rangle -|1_{i}\rangle \text{ for every }i=1,...,n
\end{equation}%
are orthogonal to any subgraph state $|G_{i}\rangle $ of $K_{n}$. Thus it
follows that $\dim (\Sigma _{K_{n}})\leq 2^{n}-n$, where $\Sigma _{K_{n}}$
is the subspace spanned by all possible subgraph states of $K_{n}$, \emph{%
i.e.} all possible graph states with $n$ vertices. To prove that the
equality holds, we have to show that the state $|D_{n}\rangle
=|00...0\rangle +\sum_{i=1}^{n}|1_{i}\rangle $ and any state with a number
of qubits in state $1$ (excitations) larger than $2$, denoted by $|\text{exc}%
_{n}\geq 2\rangle $, can be expressed as a linear combination of graph
states. This is clearly true in the simplest case of two qubits, as $%
|D_{2}\rangle \propto |++\rangle +|\vcenter{\hbox{%
\includegraphics[width=2em]{pics/bell_2}}}\rangle $ and $|11\rangle \propto
|++\rangle -|\vcenter{\hbox{\includegraphics[width=2em]{pics/bell_2}}}%
\rangle $. In order to show that it holds for generic $n$ we proceed by
induction. Suppose that for $n$ qubits it is always possible to express both
$|D_{n}\rangle $ and the states $|\text{exc}_{n}\geq 2\rangle $ as $%
\sum_{i}\alpha _{i}|G_{i}\rangle $. Then, it can be easily proved that one
can achieve both $|D_{n+1}\rangle $ and $|\text{exc}_{n+1}\geq 2\rangle $ as
follows.

Start from the state $| \text{exc}_n\ge 2 \rangle| + \rangle$, that by
hypothesis can be written as $\sum_i\alpha_i| G_i \rangle| + \rangle$. Apply
then a CZ on the qubit $n+1$ and on one of the qubits that correspond to
state $1$ in $| \text{exc}_n\ge 2 \rangle$ so that the resulting state is $%
\text{CZ}| \text{exc}_n\ge 2 \rangle| + \rangle$. Then, take the following
linear combination of the two states $| \text{exc}_n\ge 2 \rangle| + \rangle$
and $\text{CZ}| \text{exc}_n\ge 2 \rangle| + \rangle$ such that $| \text{exc}%
_{n+1}\ge 2 \rangle\propto| \text{exc}_n\ge 2 \rangle| + \rangle \pm\text{CZ}%
| \text{exc}_n\ge 2 \rangle| + \rangle$. It can be easily seen that in this
way almost all states of $n+1$ qubits with more than two excitations $|
\text{exc}_{n+1}\ge 2 \rangle$ can be created (apart from some with two
excitations that will be discussed in the following). Actually $2(2^n-n-1)$
states of the computational basis can be derived from the procedure above.
In order to generate the $n+1$ missing states (to achieve all the $%
2^{n+1}-(n+1)$ desired states) it is sufficient to start from the state $|
D_n \rangle| + \rangle=\sum_i\alpha_i| G_i \rangle| + \rangle$ (instead of $%
| \text{exc}_n\ge 2 \rangle| + \rangle$) and apply again the same reasoning.
If we now apply all possible CZ gates between the qubit $n+1$ and the rest
we can derive the state $| D_{n+1} \rangle$. If we apply a single CZ we can
achieve the $n$ missing states with two excitations (one in the qubit $n+1$
and the other in each of the $n$ qubits).

Therefore we have proved in this way that $\dim(\Sigma_{K_{n}})\ge 2^{n}-n$
and thus the equality $\dim(\Sigma_{K_{n}})= 2^{n}-n$ follows.
\end{proof}

We now introduce the following lemma that is needed for proving Theorem \ref%
{theorem::rank_of_random_graph_states}.

\begin{lemma}[Rank of a general $\protect\rho$]
Suppose that $\rho=\sum_{i=1}^D p_i| v_i \rangle\langle v_i |$ with $p_i>0$
and $\sum_{i=1}^D p_i=1$, where the states $\{| v_i \rangle\}_{i=1,...,D}$
span the space $V$ of dimension $d\le D$ (thus the set $\{| v_i
\rangle\}_{i=1,...,D}$ generally includes linearly dependent vectors). Then
the rank of $\rho$ is
\begin{equation}
\mathrm{rank}(\rho)=d.
\end{equation}
\end{lemma}

\begin{proof}
It is straightforward to see that $\mathrm{rank}(\rho)\leq d$. In order to prove that
the rank is exactly $d$, let us reason by contradiction. Suppose that there
exists $| l \rangle$ belonging to a basis $\{| j \rangle\}_{j=1,...,d}$ of $%
V $ such that $\rho| l \rangle=0$. By rewriting $| v_i \rangle=\sum_{j=1}^d
c_j^i| j \rangle$, it follows that
\begin{equation}
\rho| l \rangle=\sum_{i=1}^D p_i \sum_{j=1}^d c_j^i c_{l}^{i*}| j
\rangle=\sum_{j=1}^d \alpha_{j l}| j \rangle=0,
\end{equation}
with $\alpha_{jl} = \sum_{i=1}^D p_i c_j^i c_{l}^{i*}$. This implies that
for every $j$, $\alpha_{j l}=0$. In particular, for $j=l$ we have
\begin{equation}
\alpha_{ll}=\sum_{i=1}^D p_i |c_l^i|^2 =0.
\end{equation}
The equation above, as $p_i>0$, implies that $c_l^i=0$ for every $i$,
contradicting the hypothesis that the space $V$ has dimension $d$.
\end{proof}

\bigskip

\begin{proof}[Proof of Theorem \protect\ref%
{theorem::rank_of_random_graph_states}]
It is sufficient to apply the above lemma and Theorem \ref%
{theorem::dimension_of_complete_subgraphs_state_space} to $\rho_{K_{n}}$.
\end{proof}

\bigskip

In the following we provide an alternative proof of Theorem
\ref{theorem::dimension_of_complete_subgraphs_state_space}, via the following
lemma concerning a useful way to expand a pure state in $\Sigma_{G}$ in
terms of single qubit states.

\begin{lemma}[Expansion of states in $\Sigma_{G}$]
\label{lemma::extension_of_subgraphs_state_space} Let $|\psi\rangle=\sum_{F%
\emph{\ spans } G}c_{F}|F\rangle$ be a state in the $G$-subgraphs state
space $\Sigma_{G}$. %, while $K_{n}$ is a complete graph with $n$ vertices
Then $|\psi\rangle$ can be decomposed with respect to the bipartition
involving the single vertex $v$ as
\begin{equation}
|\psi\rangle=\frac{1}{\sqrt{2}}( | 0
\rangle_v|\phi^{0}\rangle+|1\rangle_v|\phi^{1}\rangle),
\label{eq::lemma_expansion_of_states_in_Sigma_space}
\end{equation}
with%
\begin{align}
|\phi^{0}\rangle & =\sum_{F\emph{\ spans }G}c_{F}|f_{F}\rangle,  \notag \\
|\phi^{1}\rangle & =\sum_{F\emph{\ spans }G}\sigma_{z}^{\otimes
N_{v}(F)}c_{F}|f_{F}\rangle.
\end{align}
Here $f_{F}=F-v$ is the graph achieved by removing the vertex $v$ from $F$
(and deleting all edges connected with $v$), and $N_v(F)$ is the neighborhood
of the vertex $v$.

The state $|\phi^0\rangle$ is  state  in the $(G-v)$-subgraphs
state space $\Sigma_{(G-v)}$.
\end{lemma}

\begin{proof}
Obviously, any spanning subgraph state $|F\rangle$ can be
generated by adding edges incident to the vertex $v$ to a suitable subgraph
state $|+\rangle_{v}|F-v\rangle$. In formulas, this fact can be expressed as
\begin{equation}
|F\rangle = \prod_{v_{i}\in N_{v}(F)}\text{CZ}_{v,v_{i}}%
|+\rangle_{v}|F-v\rangle
\end{equation}
Therefore, any spanning subgraph $|F\rangle$ can be rewritten as
\begin{equation}
|F\rangle=\frac{1}{\sqrt{2}}\left(  |0\rangle_{v}\otimes|f_{F}\rangle
+|1\rangle_{v}\otimes\sigma_{z}^{\otimes N_{v}(F)}|f_{F}\rangle\right)
,\label{eq.::proof_lemma_extension_of_sub_space_1}%
\end{equation}
with $|f_{F}\rangle=|F-v\rangle$. Now applying what was just found in the general
decomposition of $|\psi\rangle=\sum_{F\emph{\ spans } G}c_{F}|F\rangle$, Eq.
\eqref{eq::lemma_expansion_of_states_in_Sigma_space} follows. Since $f_{F}$
are subgraphs of $(G-v)$, the state $|\phi^{0}\rangle$ belongs to the space
$\Sigma_{(G-v)}$.
\end{proof}

\bigskip

\begin{proof}[Alternative proof of Theorem \protect\ref{theorem::dimension_of_complete_subgraphs_state_space}]
Let us first prove that $\dim(\Sigma_{K_{n}})\leq2^{n}-n$ by showing that
the $n$ mutually orthogonal states $\sigma_{z}^{v_{i}}|G_{n}^{\emptyset}%
\rangle$ ($i=1,\cdots,n$) are not in the space $\Sigma_{K_{n}}$.
In order to prove that this we reason by induction. For $n=2$, it is
trivial that there never exist coefficients $c_{\emptyset}$ and $c_{S_{2}}$
such that%
\begin{equation}
\sigma_{z}^{v_{i}}|G_{2}^{\emptyset}\rangle=c_{\emptyset}|G_{2}^{\emptyset
}\rangle+c_{S_{2}}|S_{2}\rangle.
\end{equation}
There $\sigma_{z}^{v_{i}}|G_{2}^{\emptyset}\rangle \notin \Sigma_{K_{2}}$.

We then assume that $\sigma_{z}^{v_{i}}|G_{n}^{\emptyset}\rangle$ is not in $%
\Sigma_{K_{n}}$, and want to prove this is the case for $n+1$ vertices too.
Suppose now by contradiction that $\sigma_{z}^{v_{i}}|G_{n+1}^{\emptyset
}\rangle\in\Sigma_{K_{n+1}}$, by employing Lemma \ref%
{lemma::extension_of_subgraphs_state_space} and without loss of generality,
we can then find for the first vertex $v_{1}$ that%
\begin{equation}
\sigma_{z}^{v_{1}}|G_{n+1}^{\emptyset}\rangle=\frac{1}{\sqrt{2}}\left(
|\phi^{0}\rangle|0\rangle_{v_{n+1}}+|\phi^{1}\rangle|1\rangle_{v_{n+1}}%
\right) ,
\end{equation}
with $|\phi^{0}\rangle\in\Sigma_{K_{n}}$. On the other hand, the left-hand side of the above equation is%
\begin{equation}
\sigma_{z}^{v_{1}}|G_{n+1}^{\emptyset}\rangle=\frac{1}{\sqrt{2}}\left(
\sigma_{z}^{v_{1}}|G_{n}^{\emptyset}\rangle|0\rangle_{v_{n+1}}+%
\sigma_{z}^{v_{1}}|G_{n}^{\emptyset}\rangle|1\rangle_{v_{n+1}}\right) ,
\end{equation}
which leads to%
\begin{equation}
\sigma_{z}^{v_{1}}|G_{n}^{\emptyset}\rangle=|\phi^{0}\rangle\in%
\Sigma_{K_{n}}.
\end{equation}
This contradicts the assumption that no solution exists for $n$ vertices.

In order to prove that $\dim(\Sigma_{K_{n}})\geq2^{n}-n$ we show that the
space spanned by $\Sigma_{K_{n}}$ and $\{\sigma_{z}^{v_{i}}|G_{n}^{%
\emptyset}\rangle \}_{i=1,\cdots n}$ is the full Hilbert space composed of $%
n $ qubits. To this end we prove that
\begin{align}
&\sigma_{z}^{V_{G}}|G_{n}^{\emptyset}\rangle=
\label{eq.::proof_of_theorem_dim_of_sub_space_2} \\
&\sum_{i=1}^{n}\left[ (-1)^{i}2|S_{V_G\backslash V_{i}}^{v_{i+1}}\rangle
|G_{V_{i}}^{\emptyset}\rangle-(-1)^{i}\left( \sigma_{z}^{v_{i+1}}+1\right)
|G_{n}^{\emptyset}\rangle\right] ,  \notag
\end{align}
where $V_i=\{v_1,\cdots,v_i\}$ is a set of $i$ vertices and $%
|S_{V_G\backslash V_{i}}^{v_{i+1}}\rangle$ is a star graph state on vertices
$V_G\backslash V_{i}$ and $v_{i+1}$ as the central vertex. According to
Lemma \ref{lemma::extension_of_subgraphs_state_space} we can write
\begin{equation}
|S_{n}^{v_{1}}\rangle=\frac{1}{\sqrt 2}( |0\rangle_{v_1}\otimes
|G_{n-1}^{\emptyset}\rangle+|1\rangle_{v_1} \otimes
\sigma_{z}^{V_{G}\backslash v_{1}}|G_{n-1}^{\emptyset}\rangle),
\end{equation}
and, since $|0\rangle_{v_1}|G_{n-1}^{\emptyset}\rangle =\frac{1}{\sqrt 2}
(\sigma_{z}^{v_{1}}|G_{n}^{\emptyset}\rangle +|G_{n}^{\emptyset}\rangle)$,
we can write
\begin{equation}
|1\rangle_{v_1}\otimes\sigma_{z}^{V_{G}\backslash
v_{1}}|G_{n-1}^{\emptyset}\rangle =\sqrt{2}|S_{n}^{v_1}\rangle-\frac{1}{%
\sqrt 2}(
\sigma_{z}^{v_{1}}|G_{n}^{\emptyset}\rangle+|G_{n}^{\emptyset}\rangle )
\label{eq.::proof_of_theorem_dim_of_sub_space_3}
\end{equation}
It is also easy to see that
\begin{align}
\sigma_{z}^{V_{G}}&|G_{n}^{\emptyset}\rangle \\
&=|+\rangle_{v_1}\otimes\sigma_{z}^{V_{G}\backslash
v_{1}}|G_{n-1}^{\emptyset}\rangle-\sqrt{2}|1
\rangle_{v_1}\otimes\sigma_{z}^{V_{G}\backslash
v_{1}}|G_{n-1}^{\emptyset}\rangle,  \notag
\end{align}
and, by employing Eq. (\ref{eq.::proof_of_theorem_dim_of_sub_space_3}), we
finally arrive at the following expression%
\begin{equation}
\sigma_{z}^{V_{G}}|G_{n}^{\emptyset}\rangle=-2|S_{n}^{v_{1}}\rangle+\left(
\sigma_{z}^{v_{1}}+1\right)
|G_{n}^{\emptyset}\rangle-\sigma_{z}^{V_{G}\backslash
v_{1}}|G_{n}^{\emptyset}\rangle.
\label{eq.::proof_of_theorem_dim_of_sub_space_3.5}
\end{equation}
Hence, by using Eq. (\ref{eq.::proof_of_theorem_dim_of_sub_space_3.5})
recursively we can achieve Eq. (\ref%
{eq.::proof_of_theorem_dim_of_sub_space_2}). Therefore, for any subset of
vertices $V\subseteq V_{G}$, we have that the state $\sigma_{z}^{V}|G_{n}^{%
\emptyset}\rangle$ can be expressed as a superposition of vectors in the
subspaces $\Sigma_{K_{n}}$ and $\{\sigma_{z}^{v_{i}}|G_{n}^{\emptyset}%
\rangle\}_{i=1,\cdots n}$. As the set of all vectors $\sigma_{z}^{V}|G_{n}^{%
\emptyset}\rangle$ forms the Hadamard basis, this finally proves that
\begin{equation}
\dim(\Sigma_{K_{n}})\geq 2^{n}-n.
\label{eq.::proof_of_theorem_dim_of_sub_space_4}
\end{equation}
\end{proof}

\section{Randomization overlap of some special RG states}

\label{sec::randomization_overlap_of_some_special_RG_states}

In this appendix we derive an explicit analytical result for the
randomization overlap of random star states $\rho_{S_n}^p$ and random 1D cluster states $\rho_{L_n}^p$.

\begin{solution}
Let $S_{n}$ be an $n$-vertex star graph; its randomization overlap is then
\begin{equation}
L(\rho^{p}_{S_n}) = \frac{1}{4}+\frac{3}{4}p^{n-1}.
\end{equation}
\end{solution}

\begin{proof}
The scalar product of $|S_{n}\rangle$ and any of its spanning subgraph
states $|F\rangle$ always equals $\frac{1}{4}$ (apart from the case when $|
F \rangle=| S_n \rangle$). therefore
\begin{align}
L(\rho^{p}_{S_n}) &=\frac{1}{4}\sum_{k=1}^{n-1}\binom{n-1}{k}%
p^{n-1-k}(1-p)^{k}+p^{n-1}  \notag \\
& =\frac{1}{4}+\frac{3}{4}p^{n-1}.
\end{align}
\end{proof}

\begin{figure}[t!]
\centering
\subfloat[Examples of linear clusters in $\mathcal{F}_\text{even}(n)$.]{
\includegraphics[width=0.8\linewidth]{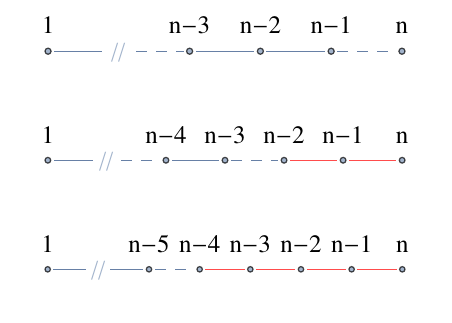}
\label{fig::F_even_of_1d_cluster}
}
\newline
\subfloat[Examples of linear clusters in $\mathcal{F}_\text{odd}(n)$.]{
\includegraphics[width=0.8\linewidth]{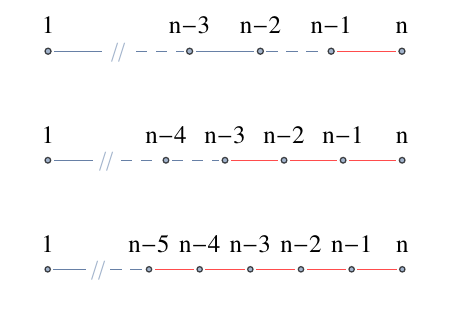}
\label{fig::F_odd_of_1d_cluster}
}
\caption{(Color online) Examples of $\mathcal{F}_\text{even}(n)$ and $%
\mathcal{F}_\text{odd}(n)$.}
\label{fig::F_even_and_F_odd_of_1d_cluster}
\end{figure}

\begin{solution}
Let $L_{n}$ be a linear cluster graph on $n$ vertices, its randomization
overlap then reads
\begin{align}
L(\rho^{p}_{L_{n}}) = &\frac{1}{\sqrt{\lambda_{p}}}\left( 1-\frac{p}{2}+%
\frac{\sqrt{\lambda_{p}}}{2}\right) \left( \frac{p}{2}+\frac{\sqrt {%
\lambda_{p}}}{2}\right) ^{n}  \label{eq.::fidelity_of_1D_RC_states} \\
& -\frac{1}{\sqrt{\lambda_{p}}}\left( 1+\frac{p}{2}+\frac{\sqrt{\lambda_{p}}%
}{2}\right) \left( \frac{p}{2}-\frac{\sqrt{\lambda_{p}}}{2}\right) ^{n},
\notag
\end{align}
with $\lambda_{p}=1-p+p^{2}$.
\end{solution}

\begin{proof}
Let us define $\mathcal{F}_{\text{even}}^{(n)}$ ($\mathcal{F}_{\text{odd}%
}^{(n)}$) as the set of spanning subgraphs of the cluster $L_{n}$ that have
paths with even (odd) number of edges connected to the last vertex $v_{n}$
(see Fig. \ref{fig::F_even_and_F_odd_of_1d_cluster} for a pictorial
explanation). The randomization overlap can thus be rewritten as
\begin{equation}
L(\rho^{p}_{L_n}) =f_{\text{even}}^{(n)}(p)+f_{\text{odd}}^{(n)}(p).
\end{equation}
where $f_{\text{even}}^{(n)}(p):=\Tr \left[ |L_{n}\rangle\langle L_n|\sum
_{F\in\mathcal{F}_{\text{even}}^{(n)}}p_{F}|F\rangle\langle F|\right] $, and
$f_{\text{odd}}^{(n)}(p):=\Tr \left[ L_n\rangle\langle L_n |\sum_{F\in%
\mathcal{F}_{\text{odd}}^{(n)}}p_{F}|F\rangle\langle F|\right] $. From the
results in Table \ref{table::results_of_scalar_product_of_example_graphs},
it is then not difficult to notice that the following recursive relations
hold
\begin{align}
f_{\text{odd}}^{(n+1)}(p) & =\frac{1-p}{4}f_{\text{even}}^{(n)}(p), \\
f_{\text{even}}^{(n+1)}(p) & =f_{\text{odd}}^{(n)}(p)+pf_{\text{even}%
}^{(n)}(p).
\end{align}
Imposing the initial conditions $f_{\text{even}}^{(2)}=p$ and $f_{\text{odd}%
}^{(2)}=(1-p)/4$, the above relations can be solved, leading to the
randomization overlap \eqref{eq.::fidelity_of_1D_RC_states}.
\end{proof}

\section{Approximation of GME witness}
\label{sec::proof_of_approximation_of_GME_witness}

%%%%%%%%%%%%%%%%% figures of isomorph_group_example %%%%%%%%%
\begin{figure}[t!]
\begin{minipage}{\linewidth}
\raggedright
\includegraphics[width=0.6\linewidth]{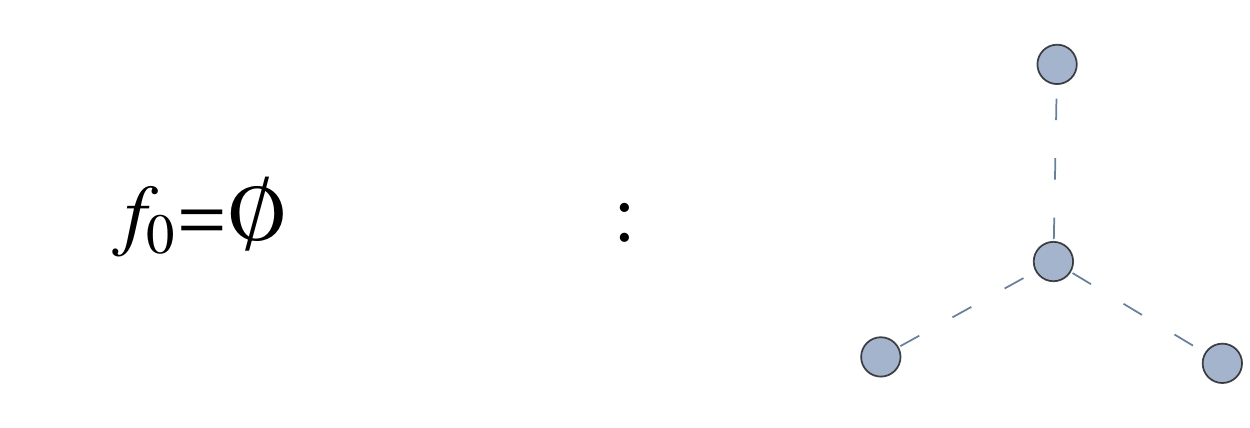}
\newline
(a) The single graph isomorphic to the empty graph
\label{fig::isomorph_group_f0_example}
\end{minipage}
\begin{minipage}{\linewidth}
\raggedright
\includegraphics[width=0.9\linewidth]{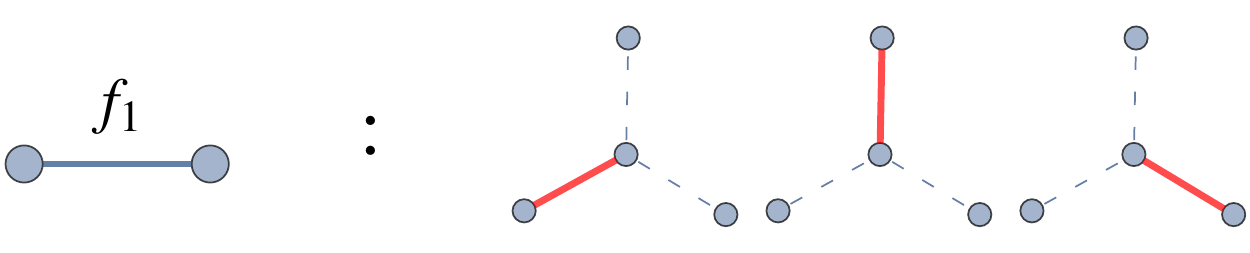}
\newline
(b) Graphs isomorphic to the $2$-vertex graph $S_2$
\label{fig::isomorph_group_f1_example}
\end{minipage}
\begin{minipage}{\linewidth}
\raggedright
\includegraphics[width=0.9\linewidth]{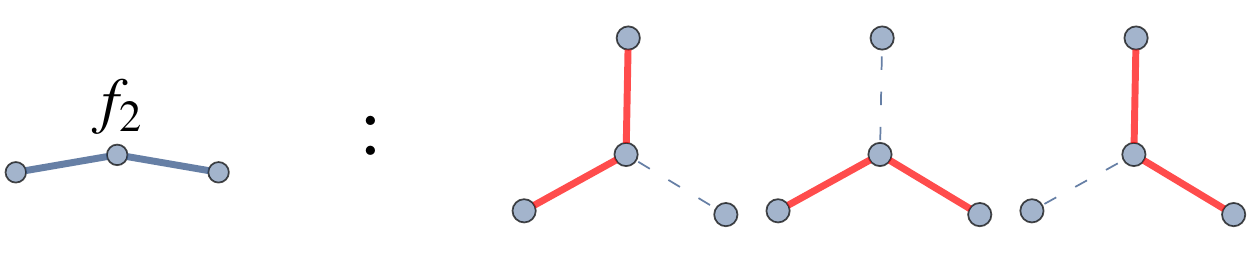}
\newline
(c) Graphs isomorphic to the star graph $S_3$ with $3$ vertices
\label{fig::somorph_group_f2_example}
\end{minipage}
\begin{minipage}{\linewidth}
\raggedright
\includegraphics[width=0.6\linewidth]{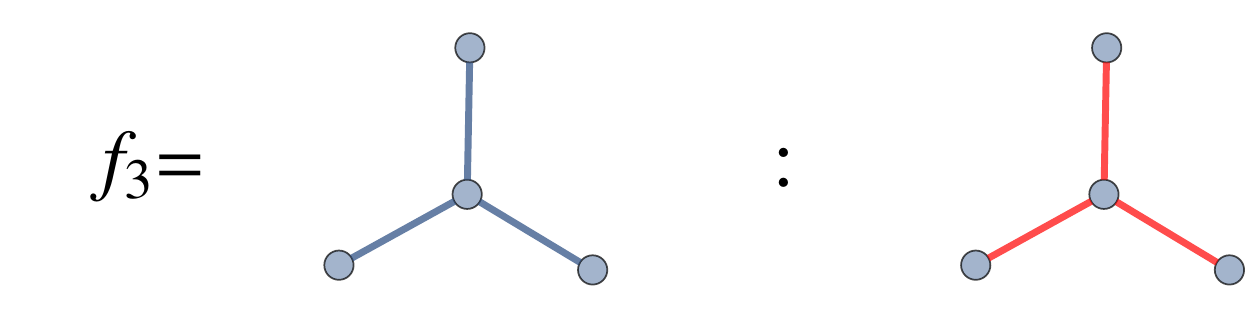}
\newline
(d) The graph isomorphic to the 4-vertex star graph $S_4$
\label{fig::isomorph_group_f3_example}
\end{minipage}
\caption{(Color online) Four different isomorphic classes of star graphs on $4$ vertices}
\label{fig::isomorph_group_example}
\end{figure}
%%%%%%%%%%%%%%%%%%%%%%%%%%%%%%%%%%%%%%%%%%%%%%%%%%%%%%%%%%%%%%%%%%

Before proving Theorem \ref{theorem::approximated_GME_witness}, it is convenient to first make the
following observation.

The randomization overlap can be easily rewritten in terms of the symmetric
difference $\tilde{F}:=F\Delta G$ as
\begin{align}
& \frac{1}{p^{|E_{G}|}}L(\rho _{G}^{p}) \\
=& \sum_{F\text{ spans }G}\left( \frac{1-p}{p}\right) ^{|E_{F\Delta G}|}\Tr%
\left[ |G^{\emptyset }\rangle \langle G^{\emptyset }|F\Delta G\rangle
\langle F\Delta G|\right] ,  \notag \\
=& \sum_{\tilde{F}\text{ spans }G}\underset{=:c_{G}^{p}(\tilde{F})}{%
\underbrace{\left( \frac{1-p}{p}\right) ^{|E_{\tilde{F}}|}\Tr\left[
|G^{\emptyset }\rangle \langle G^{\emptyset }|\tilde{F}\rangle \langle
\tilde{F}|\right] }}.  \label{eq.::contribution_in_randomization_overlap}
\end{align}%
Eq. (\ref{eq.::contribution_in_randomization_overlap}) makes it clear that
the randomization overlap can be recast as a sum of terms where any
contribution $c_{G}^{p}(\tilde{F})$ depends on both the number of edges $|E_{%
\tilde{F}}|$ and the scalar product of $|\langle G^{\emptyset }|\tilde{F}%
\rangle |$. It is clear that two isomorphic graphs $\tilde{F}%
_{1},\tilde{F}_{2}$, i.e., graphs that can be mapped into each other
by just relabelling the vertices, have the same contribution.
%$c_{G}^{p}(\tilde{F}_{1})=c_{G}^{p}(\tilde{F}_{2})$.
Therefore, it is convenient to divide the whole set of subgraphs $\tilde{F}$
into different graph-isomorphic classes (as an example, Fig. \ref%
{fig::isomorph_group_example} reports the isomorphic classes of subgraphs of
the four-vertex star graph). For values of the randomness parameter $p\geq 1/2
$, the isomorphic classes with fewer edges contribute the most to the
randomization overlap. Therefore, whenever $p\geq 1/2$ holds, it make sense
to approximate the randomization overlap as
\begin{equation}
L(\rho _{G}^{p})\geq p^{|E_{G}|}\sum_{\tilde{f}\in \mathcal{F}^{(\leq 2)}}|%
\tilde{f}|c_{G}^{p}(\tilde{f}),
\label{eq.::isomorphic_approximation_of_randomization_overlap}
\end{equation}%
where we have defined $\mathcal{F}^{(\leq 2)}:=\{\tilde{f}:|E_{\tilde{f}%
}|\leq 2\}$, i.e., any $\tilde{f}$ represents an isomorphic class of
graphs with a number of edges smaller than $2$. Notice that, since any $%
\tilde{F}\in \tilde{f}$ contributes equally, $c_{G}^{p}(\tilde{f})$ can be
regarded as $c_{G}^{p}(\tilde{F})$ in Eq. %
\eqref{eq.::contribution_in_randomization_overlap}, where $\tilde{F}$
represents any element of the class $\tilde{f}$.

%%%%%%%%%%%%%%%%% figures of isomorph_groups with edges less than 3 %%%%%%%%%

\begin{figure}[t!]
\centering
%\bigskip\bigskip
\includegraphics[width=.85\linewidth]{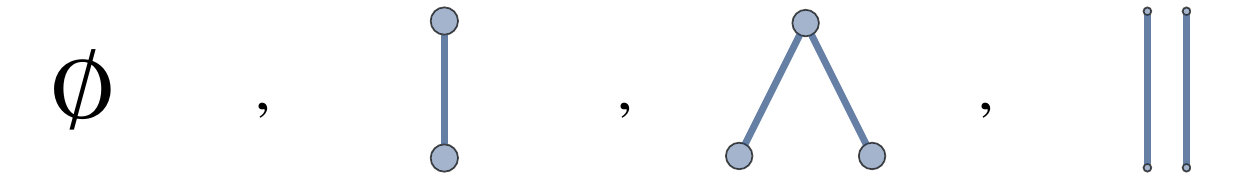}
\caption{(Color online) All the isomorphic classes $\mathcal{F}^{(\leq2)}$
with a number of edges smaller equal than $2$.}
\label{fig::isomorph_group_graphs_f012}
\end{figure}
%%%%%%%%%%%%%%%%%%%%%%%%%%%%%%%%%%%%%%%%%%%%%%%%%%%%%%%%%%%%%%%%%%%%%%%

\bigskip

We are now ready to prove Proposition \ref%
{prop::monotonicity_of_approximated_witness}, which states that the $l$%
-level approximated randomization overlap $L_{\mathcal{F}^{(\leq
l)}}(\rho_{G}^{p})$ is monotonically increasing for any $l\leq |E_G/2|$,
whenever $p\geq1/2$.

\bigskip

\begin{proof}[Proof of Proposition \protect\ref%
{prop::monotonicity_of_approximated_witness}]
Let
\begin{equation}
\lambda_{k}:=\frac{1}{\binom{|E_{G}|}{k}}\sum_{F\text{ s.t. }|E_{F}|=k}\Tr%
(|F\rangle\langle F|G\rangle\langle G|)
\end{equation}
be the average overlap $\Tr(|F\rangle\langle F|G\rangle\langle G|)$ of all
subgraphs $F$ with a fixed number of edges $k$. Since the overlap $\Tr%
(|F\rangle\langle F|G\rangle\langle G|)\leq1$, we have $\lambda_{k}\leq1$,
and thus the $l$-level approximated randomization overlap becomes%
\begin{equation}
L_{\mathcal{\ F}^{(\leq l)}}(\rho_{G}^{p})=\sum_{k=1}^{l}\lambda_{k}\binom{%
|E_{G}|}{k}p^{k}(1-p)^{|E_{G}|-k}.
\end{equation}
Now we order the indices $k$'s as follows. First we group together the
indices $k$'s that lead to the same value of the coefficients $\lambda_k$,
then we order all these sets for increasing values of the coefficients $%
\lambda_k$. In the end we get the following partition: $\{k_{1}^{(1)},%
\cdots,k_{i_{1}}^{(1)}\}$, $\{k_{1}^{(2)},\cdots,k_{i_{2}}^{(2)}\}$ ,
$\cdots$, $\{k_{1}^{(j)},\cdots,k_{i_{j}}^{(j)}\}$ where $\lambda_{k_{1}^{(1)}}=\cdots=%
\lambda_{k_{i_{1}}^{(1)}}>\lambda_{k_{1}^{(2)}}=\cdots=%
\lambda_{k_{i_{2}}^{(2)}}> \lambda_{k_{1}^{(j)}} =\cdots =
\lambda_{k_{j}^{(j)}}$. For the sake of simplicity we define $%
\lambda^{(j)}:=\lambda_{k_{1}^{(j)}}$ and $\kappa^{(j)}:=\{k_{1}^{(j)},%
\cdots,k_{i_{j}}^{(j)}\}$.

Furthermore, we need the help of the following function%
\begin{equation}
f(\kappa) =\sum_{F\text{ s.t. }|E_{F\Delta G}|\not \in \kappa\text{ }\and%
\text{ }|E_{F\Delta G}|\leq l}p_{F},
\end{equation}
which represents the probability of finding a subgraph $F$ having $k$ edges
different from $G$, where $k\leq l$ and it is not contained in $\kappa$. The
above formula can be conveniently rewritten as
\begin{align}
f(\kappa) =&\sum_{k=1}^{l}\binom{|E_{G}|}{k}p^{k}(1-p)^{|E_{G}|-k} \\
& -\sum_{k\in\kappa}\binom{|E_{G}|}{k}p^{j}(1-p)^{|E_{G}|-k}  \notag \\
=&1-\sum_{k\not \in \kappa}\binom{|E_{G}|}{k}p^{k}(1-p)^{|E_{G}|-k} \\
& -\sum_{k=l+1}^{|E_{G}|}\binom{|E_{G}|}{k}p^{k}(1-p)^{|E_{G}|-k}.  \notag
\end{align}
This function turns out to be monotonically increasing for randomness $p\ge
1/2$ and $l\leq |E_G|/2$. The $l$-level approximated
randomization overlap can be expressed in terms of functions $f(\kappa)$ as
\begin{align}
L_{\mathcal{F}^{(\leq l)}}(\rho_{G}^{p})
=&\lambda^{(1)}f(\emptyset)+(\lambda^{(2)}-\lambda^{(1)})f(\kappa^{(1)})
\notag \\
& +(\lambda^{(3)}-\lambda^{(2)})f(\kappa^{(1)}\cup\kappa^{(2)})  \notag \\
& +\cdots \\
& +(\lambda^{(j)}-\lambda^{(j-1)})f(\kappa^{(1)}\cup\cdots\cup\kappa
^{(j-1)})  \notag \\
& +(1-\lambda^{(j)})f(\kappa^{(1)}\cup\cdots\cup\kappa^{(j)}).  \notag
\end{align}
Since $(\lambda^{(i+1)}-\lambda^{(i)})>0$ and every $f(\kappa^{(1)}\cup%
\cdots\cup\kappa^{(i)})$ is monotonically increasing for randomness $p\geq
1/2$ and $l\leq |E_G|/2$, the $l$-level approximated overlap $%
L_{F^{(\leq l)}}(\rho_{G}^{p})$ is monotonically increasing whenever $p\geq
1/2$ and $l\leq |E_G|/2$.
\end{proof}

{\renewcommand{\arraystretch}{1.5}
\begin{table}[t!]
\begin{tabular}{|c|c|c|c|c|}
\hline
$\tilde{f}$ & $\emptyset$ &
\begin{tabular}{@{}c}
\\
\includegraphics[width=0.15\linewidth]{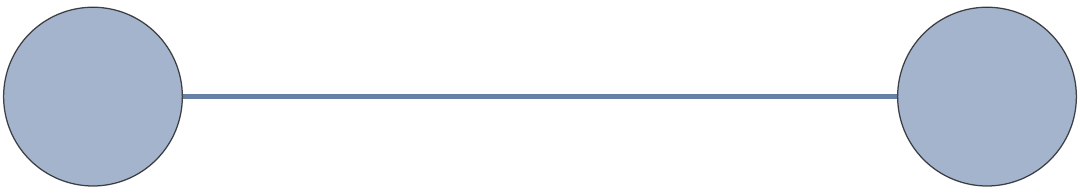} \\
$\text{Bell}_2$, $S_2$%
\end{tabular}
&
\begin{tabular}{@{}c}
\\
\includegraphics[width=0.15\linewidth]{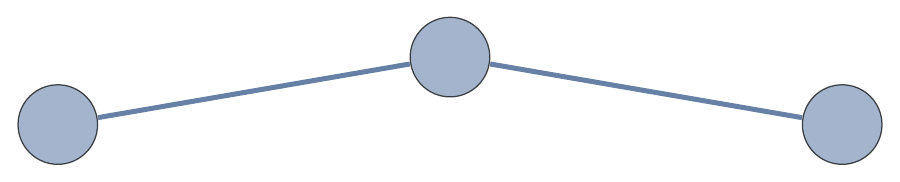} \\
$S_3$%
\end{tabular}
&
\begin{tabular}{@{}c}
\\
\includegraphics[width=0.15\linewidth]{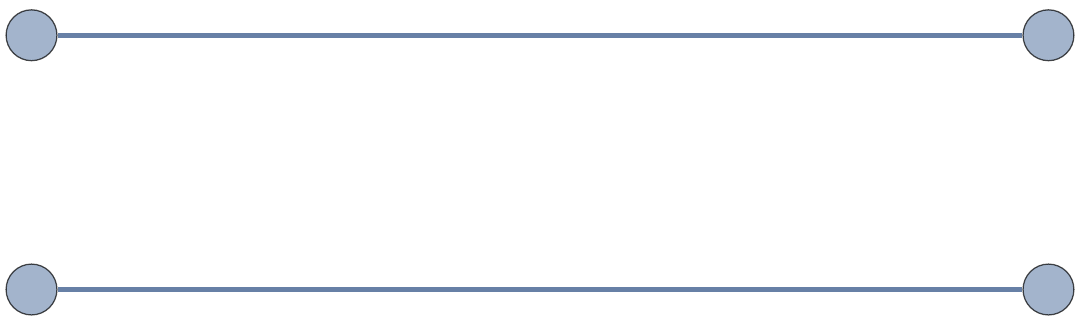} \\
$S_2 \otimes S_2$%
\end{tabular}
\\ \hline
$c_{G}^{p}(\tilde{f})$ & $1$ & $\frac{1}{4}(\frac{1-p}{p})$ & $\frac{1}{4}(%
\frac{1-p}{p})^{2}$ & $\frac{1}{16}(\frac{1-p}{p})^{2}$ \\ \hline
$|\tilde{f}|$ & $1$ & $|E_{G}|=\frac{1}{2}\sum_{v\in V}d_{v}$ &
$\sum_{v\in V}\genfrac{(}{)}{0pt}{}{d_{v}}{2}$ &
$\genfrac{(}{)}{0pt}{}{E_{G}}{2}-\sum_{v\in V}\genfrac{(}{)}{0pt}{}{d_{v}}{2}$
\\ \hline
\end{tabular}%
\caption{The cardinalities of isomorphic classes and their single element
contributions: $d_{v}$ is the vertex degree of vertex $v$ in $G$, and $E_G$
is the set of edges of $G$.}
\label{table::contributions_of_selected_isomorph_groups}
\end{table}
}

\bigskip

Finally we prove Theorem \ref{theorem::approximated_GME_witness} concerning
a possible approximation of the GME witness.

\bigskip

\begin{proof}[Proof of Theorem \protect\ref%
{theorem::approximated_GME_witness}]
The main idea of the approximation is to neglect the subgraphs of $G$ that
contain more than two edges and thus to calculate only the contribution of
the isomorphic classes of subgraphs with at most two edges (see Fig. \ref%
{fig::isomorph_group_graphs_f012}). The approximated randomization overlap
can thus be expressed as in Eq. \ref%
{eq.::isomorphic_approximation_of_randomization_overlap} and, with the help
of the results listed in Table \ref%
{table::contributions_of_selected_isomorph_groups}, can be explicitly
rewritten as
\begin{align}
L_{\mathcal{F}^{(\leq 2)}}& \left( \rho _{G}^{p}\right)  \\
& =p^{\left\vert E_{G}\right\vert }+\frac{1}{4}\left( 1-p\right)
p^{\left\vert E_{G}\right\vert -1}\left\vert E_{G}\right\vert   \notag \\
& +\frac{1}{2^{4}}\left( 1-p\right) ^{2}p^{\left\vert E_{G}\right\vert -2}%
\left[ \binom{\left\vert E_{G}\right\vert }{2}+3\sum_{v\in V_{G}}\binom{d_{v}%
}{2}\right] ,  \notag
\end{align}%
where $d_{v}$ is the degree of any vertex $v$. Since the contribution of
subgraphs with number of edges greater than $2$ is always non-negative, it
follows that $L\left( \rho _{G}^{p}\right) \geq L_{\mathcal{F}^{(\leq
2)}}\left( \rho _{G}^{p}\right) $. Therefore, we have that
\begin{equation}
I_{\mathcal{F}^{(\leq 2)}}(\rho _{G}^{p}):=1/2-L_{\mathcal{F}^{(\leq
2)}}\left( \rho _{G}^{p}\right)
\end{equation}%
is also a GME witness, in the sense that a negative value indicates the
presence of GME. Notice furthermore that $I_{\mathcal{F}^{(\leq 2)}}(\rho
_{G}^{p})\leq I_{w}(\rho _{G}^{p})$, i.e., the approximated witness is
obviously weaker than the complete one defined as $I_{w}(\rho _{G}^{p})=%
\Tr[W_G\rho_G^p]$ where $W_{G}$ is defined in Eq.\eqref{eq.::projector_witness}.

The last point of the theorem says that $p_{\mathcal{F}}\ge p_{w}$, where $%
p_{\mathcal{F}}$ ($p_{w}$) represents the threshold probability for $I_{%
\mathcal{F}^{(\leq2)}}(\rho_{G}^{p})$ ($I_{w}(\rho_{G}^{p})$), and thus it
is an upper bound for the critical probability $p_c$ also. In order to see
this, let us consider the following inequality%
\begin{align}
I_{\mathcal{F}^{(\leq2)}}(\rho_{G}^{p_{w}}) & =I_{w}(\rho_{G}^{p_{w}})+L_{%
\mathcal{F}^{(>2)}}(\rho_{G}^{p_{w}})  \notag \\
& =L_{\mathcal{F}^{(>2)}}(\rho_{G}^{p_{w}}) \\
& \geq0=I_{\mathcal{F}^{(\leq2)}}(\rho_{G}^{p_{\mathcal{F}}}),  \notag
\end{align}
where $L_{\mathcal{F}^{(>2)}}(\rho_{G}^{p_{w}})$ represents the scalar
product of $|G\rangle$ with all its subgraphs with a number of edges greater
than $2$.

Together with the fact that $I_{\mathcal{F}^{(\leq2)}}(\rho_{G}^{p})$ is a
monotonically decreasing function of $p$ for $p\ge 1/2$ (Proposition \ref%
{prop::monotonicity_of_approximated_witness}), it follows that $p_{\mathcal{F%
}}$ is always an upper bound for $p_{w}$, whenever $p\ge 1/2$. As a last
note, notice that the following chain of inequalities thus holds $p_{%
\mathcal{F}} \ge p_{w} \ge p_c$.
\end{proof}

%
%%   End of section "Appendix"
%%%%%%%%%%%%%%%%%%%%%%%%%%%%%%%%%%%%%%%%%%%%%%%%%%%%%%%%%%%%%%%%%%%%%%%%%%%%%%%%%%%%%%%%%%%%%%%%%%%%%%%%%%%%%%1

%%%%%%%%%%%%%%%%%%%%%%%%%%%%%%%%%%%%%%%%%%%%%%%%%%%%%%%%%%%%%%%%%%%%%%%%%%%%%%%%%%%%%%%%%%%%%%%%%%%%%%%%%%%%%%
%%   Bibliography
%

%using bibitem
%\begin{thebibliography}{99}
%\end{thebibliography}
%%
%using bibtex
%REVTeX 4.1 calls in a default BibTeX Style (.bst) file for each supported journal. The .bst files support displaying the titles of cited journal articles in the bibliography. To display the titles, simply use the "longbibliography" class option. Consult the REVTeX 4.1 documentation for more information.
\bibliographystyle{unsrt}

\bibliography{Random_Graph_States_and_Their_Entanglement_Properties}
%
%%   End of bibliography
%%%%%%%%%%%%%%%%%%%%%%%%%%%%%%%%%%%%%%%%%%%%%%%%%%%%%%%%%%%%%%%%%%%%%%%%%%%%%%%%%%%%%%%%%%%%%%%%%%%%%%%%%%%%%%

\end{document}